\documentclass[a4paper,UKenglish]{lipics}

\usepackage{etoolbox}
\usepackage{tikz}
\usepackage{wrapfig}
\usepackage{bbding}

\usetikzlibrary{positioning}
\usetikzlibrary{arrows}
\usetikzlibrary{calc}
\usetikzlibrary{automata}

\def\myfill{\vrule height 2mm depth 0.8mm width 0pt}
\tikzstyle{state}=[minimum size = 6.3mm, inner sep = 0.5mm, circle, draw]
\tikzstyle{statemdp}=[minimum size = 6.2mm, inner sep = 0mm, rounded corners=0, draw,rectangle split, rectangle split horizontal=false, rectangle split parts=2, rectangle split draw splits=false]

\newcommand{\bscc}[1]{\mathsf{bscc}(#1)}
\newcommand{\poly}[1]{\mathsf{poly}(#1)}
\newcommand{\symbopen}{\text{\scriptsize\FiveStarOpen}}

\newcommand{\tran}[1]{\stackrel{#1}{\rightarrow}}

\newcommand{\true}{\mathsf{true}}
\newcommand{\false}{\mathsf{false}}
\newcommand{\cyl}{\mathit{Cyl}}
\newcommand{\dist}{\mathcal{D}}
\newcommand{\mc}{\mathcal{M}}
\newcommand{\mdp}{\mathcal{M}}
\newcommand{\states}{S}
\newcommand{\actions}{A}
\newcommand{\tra}{\Delta}
\newcommand{\Gf}[1]{\mathbf{G}^{#1}}
\newcommand{\X}{\mathbf{X}\,}
\newcommand{\G}{\mathbf{G}\,}
\newcommand{\U}{\,\mathbf{U}\,}
\newcommand{\F}{\mathbf{F}\,}
\newcommand{\idf}{\mathbf{1}}

\newcommand{\Ex}[3][]{\ifblank{#1}{\mathbb{E}_{#2}}{\mathbb{E}^{#1}_{#2}}(#3)}
\newcommand{\E}[2][]{\ifblank{#1}{\mathbb{E}(#2)}{\mathbb{E}^{#1}(#2)}} 
\newcommand{\Prb}[2][]{\ifblank{#1}{\mathbb{P}_{#2}}{\mathbb{P}^{#1}_{#2}}} 
\renewcommand{\Pr}[3][]{\ifblank{#1}{\mathbb{P}_{#2}}{\mathbb{P}^{#1}_{#2}}(#3)} 
\renewcommand{\P}[2][]{\ifblank{#1}{\mathbb{P}(#2)}{\mathbb{P}^{#1}(#2)}} %

\newcommand{\ssf}[1]{x_{#1}}
\newcommand{\rabin}{R}
\newcommand{\rstates}{Q}
\newcommand{\rinit}{q_{in}}
\newcommand{\ract}{\Sigma}
\newcommand{\rtra}{\delta}
\newcommand{\racc}{\mathcal{F}}
\newcommand{\freq}[1]{\mathsf{freq}(#1)}

\newcommand{\Nset}{\mathbb{N}}

\newcommand{\sinit}{s_{in}}

\newcommand{\fix}[1]{{#1^\star}}
\newcommand{\fixs}{\fix{s}}
\newcommand{\fixq}{\fix{q}}
\newcommand{\fixI}{\fix{I}}
\newcommand{\fixh}{\fix{h}}

\newcommand{\eps}{\varepsilon}

\newcommand{\ap}[1]{\ensuremath{\mathsf{#1}}}

\allowdisplaybreaks
\newtheorem{proposition}[theorem]{\bfseries  Proposition}

\newcommand{\citeapp}[1]{Appendix~\ref{#1}}

\newcommand{\depth}{\mathit{depth}}
\newcommand{\trunc}[1]{\mathsf{trunc}(#1)}
\newcommand{\compl}[1]{\textit{co-}#1}
\newcommand{\rank}[1]{\mathit{rank}(#1)}
\newcommand{\last}[1]{\mathit{last}(#1)}
\newcommand{\thetas}{\theta_S}
\newcommand{\thetaa}{\theta_A}
\newcommand{\occu}{\mathrm{occ}}
\newcommand{\infi}{\mathrm{inf}}

\newcommand{\theoremlike}[2]{\par\medskip\penalty-250\refstepcounter{theorem}{{\bfseries\noindent#2
			\ref{#1}.}}}
\newcommand{\thmhelperpre}[2]{\theoremlike{#1}{#2}}
\newcommand{\thmhelperpost}{\par\medskip}

\newenvironment{reftheorem}[1]{\thmhelperpre{#1}{Theorem}}{\thmhelperpost}
\newenvironment{reflemma}[1]{\thmhelperpre{#1}{Lemma}}{\thmhelperpost}

\begin{document}
\title{On Frequency LTL in Probabilistic Systems}

\author[1]{Vojt\v{e}ch~Forejt}
\author[2]{Jan Kr\v{c}\'al}
\affil[1]{Department of Computer Science, University of Oxford, UK}
\affil[2]{Saarland University -- Computer Science,  Saarbr\"{u}cken, Germany}

\Copyright{Vojt\v{e}ch~Forejt, Jan Kr\v{c}\'al}



\maketitle

\begin{abstract}
We study frequency linear-time temporal logic (fLTL) which extends the
linear-time temporal logic (LTL) with a path operator $\Gf{p}$
expressing that on a path, certain formula holds with at least a given frequency $p$, thus relaxing
the semantics of the usual $\mathbf{G}$ operator of LTL. Such logic is particularly useful
in probabilistic systems, where some undesirable events such as random failures may occur
and are acceptable if they are rare enough.
Frequency-related extensions of LTL have been previously studied by several authors, where mostly
the logic is equipped with an extended ``until'' and ``globally'' operator, leading to undecidability
of most interesting problems. 

For the variant we study, we are able to establish fundamental decidability results. We show that for Markov chains,
the problem of computing the probability with which a given fLTL formula holds has the same complexity
as the analogous problem for LTL. We also show that for Markov decision processes the problem becomes
more delicate, but when restricting the frequency bound $p$ to be 1 and negations not to be outside any $\mathbf{G}^p$ operator, we can
compute the maximum probability of satisfying the fLTL formula. This can be again performed with the same time
complexity as for the ordinary LTL formulas.
\end{abstract}

\section{Introduction}

Probabilistic verification is a vibrant area of research that aims to formally check properties of stochastic systems. Among the most
prominent formalisms, with applications in e.g. modelling of network security protocols~\cite{Shm04} or
randomised algorithms~\cite{KNP12a}, are Markov chains and Markov decision processes (MDPs). Markov chains are
apt for modelling systems that contain 
purely stochastic behaviour,
for example random failures,
while MDPs can also express nondeterminism,
most commonly present as decisions of a controller or dually as adversarial events in the system.

More technically, MDP is a process that moves in discrete steps within a finite state space (labelled by sets of atomic propositions). Its evolution starts in a given initial state $s_0$. In each step a \emph{controller} chooses an action $a_i$ from a finite set $\actions(s_i)$ of actions available in the current state $s_i$. The next state $s_{i+1}$ is then chosen randomly according to a fixed probability distribution $\tra(s_i,a_i)$. The controller may base its choice on the previous evolution $s_0 a_0 \ldots a_{i-1} s_i$ and may also choose the action randomly. A Markov chain is an MDP where the set $\actions(s)$ is a singleton for each state $s$.

For the systems modelled as Markov chains or MDPs, the desired properties such as ``whenever a signal arrives to the system, the system eventually switches off'' can be often captured by
a suitable linear-time logic. The most prominent one in the verification community is Linear Temporal Logic (LTL).
Although LTL is suitable in many scenarios, it does not allow to capture some important linear-time properties,
for example that a given event takes place {\em sufficiently often}.
The need for such properties becomes even more apparent in stochastic systems, in which probabilities often model random failures.
Instead of requiring that no failure ever happens, it is natural to require that failures are infrequent, while still having the power of the LTL to specify these failures using a complex LTL formula.

A natural solution to the above problem is to extend LTL with operators that allow us to talk about {\em frequencies} of events.
Adding such operators can easily lead to undecidability as they often allow one to encode
values of potentially infinite counters~\cite{BDL-tase12,BMM14}. In both the above papers this is caused by a variant of a ``frequency until'' operator
that talks about the ratio of the number of given events happening along a finite path. The undecidability results from \cite{BDL-tase12,BMM14} carry over to the stochastic setting easily, and so, to avoid undecidability, care needs to be taken.

In this paper, we take an approach similar to~\cite{AT12} and in addition to usual operators of LTL such as $\X$, $\U$, $\G$ or $\F$ we only allow {\em frequency globally} formulae $\Gf{p} \varphi$ that require the formula $\varphi$ to hold on $p$-fraction of suffixes of an infinite path,
or more formally, $\Gf{p} \varphi$ is true on an infinite path $s_0a_0s_1 a_1 \ldots$ of an MDP if and only if
$$\liminf_{n\to\infty} \frac{1}{n}\cdot \Big|\{i \mid \text{$i < n$ and $s_i a_i s_{i+1} a_{i+1} \ldots$ satisfies $\varphi$}\}\Big| \ge p$$

This logic, which we call \emph{frequency LTL (fLTL)}, is still a significant extension to LTL, and because all operators can be nested, it allows to express much larger class of properties (a careful reader will notice that {\em nesting} of frequency operators is not the main challenge when dealing with fLTL as it can be easily removed for the price of exponential blow-up of the size of the formula).

The problem studied in this paper asks, given a Markov chain and an fLTL formula, to compute the probability with which
the formula is satisfied in the Markov chain when starting in the initial state. Analogously, for MDPs we study the
{\em controller synthesis} problem which asks to compute the maximal probability of satisfying the formula, over all controllers.

\begin{wrapfigure}[8]{r}{0.31\textwidth}
	\begin{tikzpicture}[-latex', rounded corners, xscale=1.25,yscale=1.25,
	font=\footnotesize,
	prob/.style={font=\footnotesize}]
	
	\node[statemdp] (s) at (0,0) {\myfill$s_0$\nodepart{two}\myfill\ap{w}};
	\node[statemdp] (q) at (-1.1,-1) {\myfill$s_1$\nodepart{two}\myfill\ap{q}};
	\node[statemdp] (r) at (0,-1) {\myfill$s_2$\nodepart{two}\myfill\ap{r}};
	
	\node[statemdp] (m1) at (1.5,0) {\myfill$s_3$\nodepart{two}\myfill\ap{m}};
	\node[statemdp] (mq1) at (1.5,-1) {\myfill$s_4$\nodepart{two}\myfill\ap{q}{,}\ap{m}};
	
	\draw  (q) edge  (r);
	
	\draw  (s) -| node[pos=0.3, below] {w} node[pos=.8,left,prob] {0.5}  (q);
	\draw  (s) -- (0.6,0) coordinate (v2){} node[pos=.7,above] {m} -- node[pos=.7,above,prob] {0.5}   (m1);
	\draw  (m1)  |- (1,0.5) -| (s);
	
	\draw  (mq1)  -- +(-0.6,0) -- (s);

	\draw  ($(s)+(-0.4,0.4)$) -- (s);
	
	\draw  (s) -- +(0.6,-0)  -- node[pos=.7,right,prob] {0.5} (mq1);
	\draw  (s) edge[out=180,in=225,loop, looseness=7] node[pos=0.8,below,prob] {0.5}  (s);
	\draw  (r) edge (s);
	\end{tikzpicture}
	\caption{An example MDP.\label{fig:intro}}
\end{wrapfigure}
For an example of possible application, suppose a network service accepts \ap{queries} by immediately sending back $\ap{responses}$, and in addition
it needs to be switched off for $\ap{maintenance}$ during which the queries are not accepted.
In most states, a new \ap{query} comes in the next step with probability $0.5$. In the \ap{waiting} state, the system chooses either to wait further (action $w$), or to start a \ap{maintenance} (action $m$) which takes one step to finish.
The service is modelled as an MDP from Figure~\ref{fig:intro}, leaving some parts of the behaviour unspecified. The aim is to synthesise a control strategy that meets with a given probability the requirements on the system. Example requirements can be given by a formula $ \G\F \ap{m} \;\land\; \G\F (\ap{q} \rightarrow \X \ap{r} )$ which will require that the service sometimes accepts the request, and sometimes goes for maintenance. However, there is no quantitative restriction on how often the maintenance can take place, and such restriction is inexpressible in LTL. However, in fLTL we can use the formula $ \G\F \ap{m} \;\land\; \Gf{0.95} (\ap{q} \rightarrow \X \ap{r} )$ to restrict
that the service is running sufficiently often, or a strong restriction $\G\F \ap{m} \;\land\; \Gf{1} (\ap{q} \rightarrow \X \ap{r} )$ saying that it is running with frequency 1. The formula may also contain several frequency operators.
In order to push the frequency of correctly handled queries towards a bound $p$, the controller needs to choose to perform the maintenance less and less frequently during operation.

\subparagraph{Related work}
Controller synthesis for ordinary LTL is a well-studied problem solvable in time polynomial in the size of the model and doubly exponential
in the size of the formula~\cite{BP08}. Usually, the LTL formula is transformed to an equivalent Rabin automaton, and
the probability of reaching certain subgraphs is computed in a product of the MDP (or Markov Chain) with the automaton.

A similar approach is taken by~\cite{AT12}. They study a logic similar to our fLTL, where LTL is extended with a mean-payoff reward constraints
in which the reward structures are determined by validity of given subformulas. The authors show that any formula can be converted
to a variant of non-deterministic B\"uchi automata, called multi-threshold mean-payoff B\"uchi automata, with
decidable emptiness problem, thus yielding decidability for model-checking and
satisfiability problems of labelled transition systems. Results of~\cite{AT12} cannot be applied
to probabilistic systems:
here one needs to work with {\em deterministic automata},
but as pointed out in \cite[Section 4, Footnote 4]{AT12} the approach of~\cite{AT12} heavily relies on non-determinism, since reward values depend on complete future, and so deterministic ``multi-threshold mean-payoff Rabin automata'' are strictly less expressive than the logic.
Another variant of frequency LTL was studied in~\cite{BDL-tase12,BMM14}, in which also a modified until operator is introduced. 
The work~\cite{BDL-tase12} maintains boolean semantics of the logic, while in~\cite{BMM14} the value of a formula is a number between
0 and 1. Both works obtain undecidability results for their logics, and~\cite{BDL-tase12} also yields decidability for restricted nesting.
Another logic that speaks about frequencies on a finite interval was introduced in~\cite{DBLP:journals/corr/abs-1111-3111} but provides analysis algorithm only for a bounded fragment.

Significant attention has been given to the study of quantitative objectives. The work \cite{boker2011temporal} adds mean-payoff objectives to temporal logics, but only as atomic propositions and not allowing more complex properties to be quantified. 
The work \cite{baier2014weight} extends LTL with another form of quantitative operators, allowing accumulated weight constraint expressed using automata, again not allowing quantification over complex formulas.
\cite{bloem2009better} introduces lexicographically ordered mean-payoff objectives in non-stochastic parity games and \cite{chatterjee2011energy} gives a polynomial time algorithm for almost-sure winning in MDPs with mean-payoff and parity objectives. These objectives do not allow to attach mean-payoff (i.e. frequencies) to properties more complex than atomic propositions. The solution to the problem requires infinite-memory strategy which at high level has a form similar to the form of strategies we construct for MDPs. Similar strategies also occur in \cite{chatterjee2005mean,chatterjee2012games,BFK08} although each of these works deals with a fundamentally different problem.

In branching-time logics, CSL is sometimes equipped with a ``steady-state'' operator whose semantics is similar
to our $\Gf{p}$ (see e.g. \cite{BHHK00}), and an analogous approach has been taken for the logic PCTL~\cite{kuvcera2005controller,de1998specify}.
In such logics every temporal subformula
is evaluated over states, and thus the model-checking of a frequency operator can be directly reduced to achieving a single mean-payoff reward.
This is contrasted with our setting in which the whole formula is evaluated over a single path, giving rise to much more complex
behaviour.

\subparagraph{Our contributions}
To our best knowledge, this paper gives the first decidability results for probabilistic verification against linear-time temporal logics extended by frequency operators with \emph{complex nested subformulas} of the logic.

We first give an algorithm for computing the probability of satisfying an fLTL formula in a Markov Chain.
The algorithm works by breaking the fLTL formula into linearly many ordinary LTL formulas, and then off-the-shelf
verification algorithms can be applied. We obtain that the complexity of fLTL model-checking
is the same as the complexity of LTL model checking. Although the algorithm itself is very simple, some
care needs to be taken when proving its correctness: as we explain later, the ``obvious'' proof approach would fail
since some common assumptions on independence of events are not satisfied.

We then proceed with Markov decision processes, where we show that the controller synthesis problem is significantly more complex.
Unlike the ordinary LTL, for fLTL the controller-synthesis problem may require strategies to use {\em infinite memory}, even
for very simple formulas. On the positive side, we give an algorithm for synthesis of strategies for formulas in which the
negations are pushed to atomic propositions, and all the frequency operators have lower bound 1. Although this might appear to be
a simple problem, it is not easily reducible to the problem for LTL, and the proof of the correctness of the algorithm is in
fact very involved.
This is partly because even if a strategy satisfies the formula, it can exhibit a very ``insensible'' behaviour,
as long as this behaviour has zero frequency in the limit. In the proof, we need to identify these cases and eliminate them.
Ultimately, our construction again yields the same complexity as the problem for ordinary LTL. We believe the contribution of
the fragment is both practical, as it gives a ``weaker'' alternative of the $\G$ operator usable in controller synthesis, and
theoretical, giving new insights into many of the challenges one will face in solving the controller-synthesis problem for the whole
fLTL.

\section{Preliminaries}
\label{sec:prelims}

We now proceed with introducing basic notions we use throughout this paper.

A {\em probability distribution} over a finite or countable set $X$ is a function $d : X \to [0,1]$ such that
$\sum_{x\in X} d(x) = 1$, and $\dist(X)$ denotes the set of all probability distributions over $X$.

\subparagraph{Markov decision processes and Markov chains}
A {\em Markov decision process} (MDP) is a tuple $\mdp = (\states, \actions, \tra)$ where $\states$ is a finite set of states,
$\actions$ is a finite set of actions, and $\tra:\states\times\actions\rightarrow \dist(S)$
is a partial probabilistic transition function.
A {\em Markov chain} (MC) is an MDP in which for every $s\in \states$ there is exactly one $a$ with
$\tra(s,a)$ being defined. We omit actions completely when we speak
about Markov chains and no confusion can arise.

An infinite {\em path}, also called \emph{run}, in $\mdp$ is a sequence $\omega = s_0a_0s_1a_1\cdots$ of states and actions
such that $\tra(s_i,a_i)(s_{i+1}) > 0$ for all $i$, and we denote by $\omega(i)$ the suffix $s_ia_is_{i+1}a_{i+1}\cdots$.
A finite path $h$, also called \emph{history}, is a prefix of an infinite path
ending in a state. 
Given a finite path $h=s_0a_0s_1a_1\cdots s_i$ and a finite or infinite path $h'=s_ia_is_{i+1}a_{i+1}\cdots$
we use $h\cdot h'$ to denote the concatenated path $s_0a_0s_1a_1\cdots$.
The set of paths starting with a prefix $h$ is denoted by $\cyl(h)$, or simply by $h$ if it leads to no confusion.
We overload the notation also for sets of histories, we simply use $H$ instead of $\bigcup_{h\in H} \cyl(h)$.

A {\em strategy} is a function $\sigma$ that to every finite path $h$ assigns a probability distribution over actions such that
if an action $a$ is assigned a non-zero probability, then $\tra(s,a)$ is defined where $s$ denotes the last state in $h$. A strategy $\sigma$ is {\em deterministic} if it assigns
Dirac distribution to any history, and {\em randomised} otherwise. Further, it is {\em memoryless} if its choice only depends on
the last state of the history, and {\em finite-memory} if there is a finite automaton such that $\sigma$ only makes its choice based on
the state the automaton ends in after reading the history.

An MDP $\mdp$, a strategy $\sigma$ and an initial state $s_{in}$ give rise to a probability space $\Prb[s_{in}]\sigma$ defined in a standard way~\cite{KSK76}. For a history $h$ and a measurable set of runs $U$ starting from the last state of $h$, we denote by $\Prb[h]{\sigma}(U)$ the probability
$\Prb[s_{in}]{\sigma}(\{h\cdot\omega \mid \omega \in U \} \mid h)$.
Similarly, for a random variable $X$ we denote by $\Ex[s_{in}]{\sigma}{X}$ the expectation of $X$ in this probability space and by $\Ex[h]{\sigma}{X}$ the expectation $\Ex[s_{in}]{\sigma}{X_h \mid h}$. Here, $X_h$ is defined by $X_h(h\cdot\omega) = X(\omega)$ for runs of the form $h\cdot \omega$, and by $X_h(\omega') = 0$ for all other runs.
We say that a property holds {\em almost surely} (or for almost all runs, or almost every run)
if the probability of the runs satisfying the property is $1$.

\subparagraph{Frequency LTL}
The syntax of {\em frequency LTL (fLTL)} is defined by the equation:
\[
\varphi ::= \alpha \mid \neg \varphi \mid \varphi \vee \varphi \mid \X \varphi \mid \varphi \U \varphi \mid \Gf{p} \varphi
\]
where $\alpha$ ranges over a set $\mathit{AP}$ of atomic propositions. The logic LTL is obtained by omitting the rule for
$\Gf{p} \varphi$.
For Markov chains we study the whole fLTL whereas for MDP, we restrict to a fragment that we call \emph{1-fLTL}. In this fragment, negations only occur immediately preceding atomic propositions, and $\Gf{p}$ operators occur only with $p=1$.

For an infinite sequence
$\gamma=x_1x_2\ldots$ of numbers, we set $\freq{\gamma} := \liminf_{i\to \infty} \frac{1}{i} \sum_{j=1}^i x_i$.
Given a valuation $\nu : S \rightarrow 2^\mathit{AP}$, the semantics of fLTL is defined over a path $\omega = s_0a_0s_1\dots$ of an MDP as follows.
\[
{%
	\begin{array}{llll}
	\omega\models \alpha &\text{iff } \alpha\in \nu(s_0)&
	\ \omega\models \X \varphi &\text{iff } \omega(1) \models\varphi\\
	\omega\models \neg \varphi &\text{iff } \omega \not\models\varphi&
	\ \omega\models \varphi_1\U\varphi_2 &\text{iff } \exists k : \omega(k) \models\varphi_2 \wedge \forall \ell{<}k : \omega(\ell) \models \varphi_1\\
	\omega\models \varphi_1 {\vee} \varphi_2 &\text{iff } \omega {\models}\varphi_1\text{ or }\omega{\models}\varphi_2&
	\ \omega\models \Gf{p}\varphi &\text{iff } \freq{\idf_{\varphi,0}\idf_{\varphi,1}\ldots} \ge p\\
	\end{array}
}
\]
where $\idf_{\varphi,i}$ is $1$ for $\omega$ iff $\omega(i) \models \varphi$, and $0$ otherwise.
We define $\true, \false$, $\wedge$, and $\rightarrow$ by their usual definitions and introduce standard operators
$\F$ and $\G$ by putting $\F\varphi \equiv \true \U \varphi$ and
$\G\varphi \equiv \neg \F \neg \varphi$. 
Finally, we use $\Pr{\sigma}{\varphi}$ as a shorthand for $\Pr{\sigma}{\{\omega \mid \omega\models \varphi\}}$.

\begin{definition}[Controller synthesis]
	The controller synthesis problem asks to decide, given an MDP $\mdp$, a valuation $\nu$, an initial state $\sinit$, an fLTL formula $\varphi$ and
	a probability bound $x$, whether
	$\Pr[\sinit]{\sigma}{\varphi} \ge x$ for some strategy $\sigma$.
\end{definition}

As an alternative to the above problem, we can ask to compute the maximal possible $x$ for which the answer is true.
In the case of Markov chains, we speak about {\sf Satisfaction} problem since there is no strategy to synthesise.

\subparagraph{Rabin automata}
A {\em (deterministic) Rabin automaton} is a tuple $\rabin = (\rstates, \rinit, \ract, \rtra, \racc)$ where $\rstates$
is a finite set of states, $\ract$ is an input alphabet, $\rtra : \rstates \times \ract \rightarrow \rstates$
is a transition function, and $\racc \subseteq \rstates \times \rstates$ is an accepting condition.
A computation of $\rabin$ on an infinite word $\varrho = a_0a_1\ldots$ over the alphabet $\ract$ is the
infinite sequence $\rabin[\varrho] = q_0q_1\ldots$ with $q_0=\rinit$ and $\rtra(q_i,a_i) = q_{i+1}$. A computation
is accepting (or ``$R$ accepts $\varrho$'')
if there is $(E,F) \in \racc$ such that all states of $E$ occur only
finitely many times in the computation, and some state of $F$ occurs in it infinitely many times.
For a run $\omega=s_0a_0s_1a_1\ldots$ and a valuation $\nu$, we use $\nu(\omega)$ for the sequence $\nu(s_0)\nu(s_1)\ldots$
of sets of atomic propositions.

As a well known result~\cite{BP08}, 
for every MDP $\mdp$, valuation $\nu$ and an LTL formula $\varphi$ there is a Rabin automaton $\rabin$ over the alphabet $2^{\mathit{AP}}$ such that $\rabin$ is constructible in doubly exponential time and $\omega\models \varphi$ iff $\rabin$
accepts $\nu(\omega)$. We say that $\rabin$ is equivalent to $\varphi$.
It is not clear whether this result and the definition of Rabin automata can be extended to work with fLTL in a way
that would be useful for our goals. The reason for this is, as pointed out in~\cite[Section 4, Footnote 4]{AT12}, that the frequencies in fLTL depend on the future of a run, and so require non-determinism, which is undesirable in stochastic verification.

\section{Satisfaction problem for Markov Chains}
\label{sec:mc}

In this section we show how to solve the satisfaction problem for MCs and fLTL.
Let us fix a MC $\mc=(S,\Delta)$, an initial state $\sinit$ and fLTL formula $\psi$.
We will use the notion of {\em bottom strongly connected component} (bscc) of $\mc$, which
is a set of
states $S'$ such that for all $s \in S'$ the set of states reachable from $s$
is exactly $S'$. If $s$ is in a bscc, by $\bscc{s}$ we denote the bscc containing $s$.

We first describe the algorithm computing the probability of satisfying $\psi$ from $\sinit$, and
then prove its correctness.

\subparagraph{The algorithm}
The algorithm proceeds in the following steps. First, 
for each state contained in some bscc $B$, we compute the steady-state
frequency $x_s$ of $s$ within $B$. It is the number $\E[s]{\freq{\idf_{s,0}\idf_{s,1}\ldots}}$ where $\idf_{s,i}(\omega)$ equals $1$ if the $i$th state of $\omega$ is $s$, and $0$, otherwise. Afterwards,
we repeat the following steps and keep modifying $\psi$ for as long as it contains any $\Gf{p}$ operators:
\begin{enumerate}
	\item Let $\varphi$ be a LTL formula and $p$ a number such that $\psi$ contains $\Gf{p}\varphi$.
	\item Compute $\P[s]{\varphi}$ for every state $s$ contained in some bscc.
	\item Create a fresh atomic proposition $\alpha_{\varphi,p}$ which is true in a state $s$ iff $s$ is contained in a bscc and
	$\sum_{t\in \bscc{s}} \ssf{t}\cdot \P[t]{\varphi} \ge p$. 
	\item Modify $\psi$ by replacing any occurrence of $\Gf{p}\varphi$ with $\F \alpha_{\varphi,p}$.
\end{enumerate}
Once $\psi$ contains no $\Gf{p}$ operators, it is an LTL formula and we can use off-the-shelf techniques to compute $\P[\sinit]{\psi}$,
which is our desired value.

\subparagraph{Correctness}
The correctness of the algorithm relies on the fact that $\alpha_{\varphi,p}$ labels states in a bscc $B$ if almost every run reaching $B$ satisfies the corresponding frequency constraint:

\begin{proposition}\label{prop:mc-replace}
	For every LTL formula $\varphi$, every number $p$, every bscc $B$ and almost every run $\omega$ that enters $B$ we have
	$\omega \models \Gf{p}{\varphi}$ if and only if $\sum_{t\in B} \ssf{t}\cdot \P[t]{\varphi} \ge p$.
\end{proposition}

The proposition might seem ``obviously true'', but the proof is not trivial.
The main obstacle is that satisfactions of $\varphi$ on $\omega(i)$ and $\omega(j)$
are {\em not independent} events in general: for example if $\varphi \equiv \F \alpha$ and $i<j$, then $\omega(j)\models \varphi$ implies
$\omega(i)\models \varphi$. Hence we cannot apply the Strong law of large numbers (SLLN) for independent random variables or
Ergodic theorem for Markov chains~\cite[Theorems 1.10.1-2]{Norris}, which would otherwise be obvious candidates.
Nevertheless, we can use the following variant of SLLN for correlated events.
\begin{lemma}\label{lemma:alt-slln}
	Let $Y_0,Y_1\ldots$ be a sequence of random variables which only take values $0$ or $1$ and have expectation $\mu$. Assume there are $0<r,c<1$ such that for all $i,j\in \Nset$ we have $\E{(Y_i-\mu)(Y_j - \mu)} \le r^{\lfloor c |i-j|\rfloor}$. Then
	$\lim_{n\rightarrow \infty} \sum_{i=0}^n Y_i / n = \mu$ almost surely.
\end{lemma}

Using the above lemma, we now prove Proposition~\ref{prop:mc-replace} for fixed $\varphi$, $B$, $p$. Let $\rabin$ denote the Rabin automaton equivalent to $\varphi$ and $\mc \times \rabin$ be the Markov chain product of $\mc$ and $\rabin$. 

First, we say that a finite path $s_0\ldots s_k$ of $\mc$ is \emph{determined} if the state $q_k$ reached by $R$ after reading $\nu(s_0 \ldots s_{k-1})$ satisfies that $(s_k,q_k)$ is in a bscc of $\mc \times \rabin$.
We point out that for a determined path $s_0\ldots s_k$, either almost every run of $\cyl(s_0\ldots s_k)$ satisfies $\varphi$, or almost no run of $\cyl(s_0\ldots s_k)$ satisfies $\varphi$.
Also, the probability of runs determined within $k$ steps is at least
$\sum_{i=0}^{\lfloor k/M\rfloor} (1-r^M)^{i} r^M = 1-(1-r^M)^{\lfloor k/M\rfloor}$ where $M$ is the number of states of $\mc \times \rabin$ and $r$ is the minimum probability that occurs in $\mc \times \rabin$.

Now fix a state $s \in B$.
For all $t\in B$ and $i\ge 0$
we define random variables $X^t_i$ over runs initiated in $s$. We let $X^t_i(\omega)$ take value $1$ if $t$ is visited at least $i$ times in $\omega$ and the suffix of $\omega$ starting from the $i$th visit to $t$ satisfies $\varphi$. Otherwise, we let $X^t_i(\omega)=0$.
Note that all $X^t_i$ have a common expected value $\mu_t = \P[t]{\varphi}$.

Next, let $i$ and $j$ be two numbers with $i\le j$.
We denote by $\Omega$ the set of all runs and by $D$ the set of runs $\omega$ for which the suffixes starting from the $i$th visit to $t$ are determined before the $j$th visit to $t$ (note that $D$ can possibly be $\emptyset$).
Because on these determined runs $\E[s]{X^t_j - \mu_t \mid D}=0$, we get
\[
\E[s]{(X^t_i - \mu_t) (X^t_j - \mu_t)}
\le 1-\P[s]{D}
\le (1-r^M)^{\lfloor (i-j)/M\rfloor}
\]
as shown in Appendix~\ref{app:rv-xti}. 
Thus, Lemma~\ref{lemma:alt-slln} applies to the random variables $X^t_i$ for a fixed $t$.
Considering all $t\in B$ together, we show in Appendix~\ref{app:mc-bscc}
that
$
\freq{\idf_{\varphi,0}\idf_{\varphi,1}\ldots} =  
\sum_{t \in\bscc{s}} \ssf{t} \P[s]{\varphi}
$
for almost all runs initiated in the state $s$ we fixed above.
Because almost all runs that enter $B$ eventually visit $s$, and because satisfaction of  $\Gf{p}\varphi$ is independent of any prefix,
the proof of Proposition~\ref{prop:mc-replace} is finished, and we can establish the following.

\begin{theorem}
	The satisfaction problem for Markov chains and fLTL is solvable in time polynomial in the size of the model, and doubly exponential in the size of the formula. 
\end{theorem}

\section{Controller synthesis for MDPs}
\label{sec:mdp}

We now proceed with the controller synthesis problem for MDPs  and 1-fLTL. The problem for this restricted fragment of 1-fLTL is still highly non-trivial. In particular, it is not equivalent to synthesis for the \emph{LTL} formula where every $\Gf{1}$ is replaced with $\G$. Indeed, for satisfying any LTL formula, finite memory is sufficient, while for 1-fLTL, the following theorem shows that infinite memory may be necessary.

\begin{theorem}\label{thm:inf}
	There is a 1-fLTL formula $\psi$ and a Markov decision process $\mdp$ with valuation $\nu$ such that the answer to the controller synthesis problem is ``yes'', but there is no finite-memory strategy witnessing this.
\end{theorem}
\begin{proof}[Proof idea]
	Consider the MDP from Figure~\ref{fig:intro} together with the formula
	$\psi = \G\F \ap{m} \;\land\; \Gf{1} (\ap{q} \rightarrow \X \ap{r} )$.
	Independent of the strategy being used, no run initiated in $s_4$ satisfies
	the subformula $\ap{q} \rightarrow \X \ap{r}$, while every run initiated in any other state satisfies
	this subformula. This means that we need the frequency of visiting $s_4$ to be $0$.
	The only finite-memory strategies achieving this are those that from some history
	on never choose to go right in the controllable state. However, under such strategies the formula  
	$\G\F \ap{m}$ is not almost surely satisfied.
	On the other hand, the infinite-memory strategy that on $i$-th visit to $s_0$ picks $\ap{m}$ if and only if $i$ is of the form $2^j$ for some $j$ satisfies $\psi$.
	
	Note that although the above formula requires infinite memory due to ``conflicting'' conjuncts, infinite memory is needed already for simpler
	formulae of the form $\Gf{1}(\ap{a} \U \ap{b})$.
\end{proof}

The above result suggests that it is not possible to easily re-use verification algorithms for
ordinary LTL. Nevertheless, our results allow us to establish the following theorem.

\begin{theorem}\label{thm:main-mdp}
	The controller-synthesis problem for MDPs and 1-fLTL is solvable in time polynomial
	in the size of the model and doubly exponential in the size of the formula.
\end{theorem}

For the rest of this section, in which we prove Theorem~\ref{thm:main-mdp}, we fix an MDP $\mdp$, valuation $\nu$, an initial state $\sinit$,
and a 1-fLTL formula $\psi$. The proof is given in two parts. In the first part, in Section~\ref{sec:strat} we show that the controller-synthesis problem is equivalent to problems of reaching a certain set $\Upsilon$ and then ``almost surely winning'' from this set. To prove this, the ``almost surely winning'' property will further be reduced to finding certain set of states and actions on a product MDP (Lemma~\ref{prop:screwing-paths}). In the second part of the proof, given in Section~\ref{sec:alg}, we will show that all the above sets can be computed.

\subsection{Properties of satisfying strategies}
\label{sec:strat}

Without loss of generality suppose that the formula $\psi$ does not contain $\Gf{1}$ as the outermost operator, and that it contains $n$ subformulas of the form $\Gf{1} \varphi$. Denote these subformulas $\Gf{1}\varphi_1,\ldots \Gf{1}\varphi_n$. 
For example, $\psi = \ap{i} \rightarrow \big( \G (\ap{q} \rightarrow \ap{a}) \,\land\, \Gf{1} ( \ap{p_1} \U \ap{r} \lor \Gf{1} \ap{a} ) \big)
$
contains $\varphi_1= \ap{p_1} \U \ap{r} \lor \Gf{1} \ap{a}$ and $\varphi_2 = \ap{a}$.

The first step of our construction is to convert these formulae $\psi, \varphi_1,\ldots,\varphi_n$ to equivalent Rabin automata. However, as the formulae contain $\Gf{1}$ operators, they cannot be directly expressed using Rabin automata (and as pointed out by \cite{AT12}, there is a fundamental obstacle preventing us from extending Rabin automata to capture $\Gf{p}$). 

To overcome this, we replace all occurrences of $\Gf{1}\varphi_i$ in such formulae by either $\true$ or $\false$, to capture that the frequency constraint is or is not satisfied on a run. 
Such a replacement can be fixed only after a point in the execution is reached where it becomes clear which frequency constraints in $\psi$ can be satisfied.
For a formula $\xi\in \{\psi, \varphi_1,\ldots \varphi_n\}$, any subset $I \subseteq \{1, \ldots,n \}$ of satisfied constraints defines a LTL formula $\xi^I$ obtained from 
$\xi$ by replacing all subformulas $\Gf{1}\varphi_i$ (not contained in any other $\Gf{1}$) with $\true$ if $i\in I$ and with $\false$ if $i\not\in I$. The Rabin automaton for $\xi^I$ is then denoted by $\rabin_{\xi,I}$.
For the formula $\psi$ above, we have, e.g.,
$\psi^{\{1\}} = \psi^{\{1,2\}} = \ap{i} \rightarrow \big( \G (\ap{q} \rightarrow \ap{a}) \,\land\, \true \big)$, and $\varphi_1^{\emptyset} = \ap{p_1} \U \ap{r} \lor \false$. 

We use $Q$ for a disjoint union of the state spaces of these distinct Rabin automata, and $Q_\psi$ for a disjoint union of the state spaces of the automata $\rabin_{\psi,I}$, called \emph{main} automata, for all $I$. 
Finally, for $q\in Q$ belonging to a Rabin automaton $\rabin$ we denote by $\rabin^q$ the automaton obtained from $\rabin$ by changing its initial state to $q$.

Let us fix a state $s$ of $\mdp$ and a state $q$ of $R_{\psi,I}$ for some $I \subseteq \{1,\ldots,n\}$. We say that a run $s_0a_0s_1a_1\ldots$ \emph{reaches} $(s,q)$ if for some $k$ we have $s = s_k$ and $q$ is the state reached by the main automaton $\rabin_{\psi,I}$ after reading $\nu(s_0a_0s_1\ldots s_{k-1})$. Once $(s,q)$ is reached, we say that a strategy $\sigma'$ is \emph{almost-surely winning} from $(s,q)$ if $\Prb[s]{\sigma'}$ assigns probability $1$ to the set of runs $\omega$ such that $\nu(\omega)$ is accepted by $\rabin_{\psi,I}^q$, and $\omega\models G^{1} \varphi_i$ whenever%
\footnote{Note that the product construction that we later introduce does not give us ``iff'' here. This is also why we require the negations to only occur in front of atomic propositions} we have $i\in I$.

\begin{proposition}\label{prop:reach}
	There is a strategy $\sigma$ such that $\Pr{\sigma}{\psi}=x$
	if and only if
	there is a set $\Upsilon \subseteq S\times Q_\psi$ for which the following two conditions are satisfied:
	\begin{enumerate}
		\item\label{item:2a} There is a strategy $\sigma'$ such that $\Pr{\sigma'}{\{ \omega \mid \text{$\omega$ reaches a pair from $\Upsilon$} \}} = x$.
		\item\label{item:2b} For any $(s,q) \in \Upsilon$ there is $\sigma_{s,q}$ almost-surely winning from $(s,q)$.
	\end{enumerate}
\end{proposition}

\noindent
Intuitively, the proposition relies on the fact that if $\Gf{1} \varphi_i$ holds on a run, then it holds on all its suffixes, and says that any strategy $\sigma$ can be altered so that almost surely there will be a prefix after which we know which of the $\Gf{1} \varphi_i$ will be satisfied.

\begin{example}\label{ex:upsilon-intuition}
	Let us first illustrate the set $\Upsilon$ on a formula $ \X q \;\land\; \G\F \ap{m} \;\land\; \Gf{1} (\ap{q} \rightarrow \X \ap{r} )$ that can be satisfied on the MDP from Figure~\ref{fig:intro} with probability $0.5$. Figure~\ref{fig:rabin} shows Rabin automata for the formulae
	$\psi^{\{1\}} = \X \ap{q} \land \G\F \ap{m} \land \true$ (left) and $\varphi_1^{\{1\}} = \ap{q} \rightarrow \X\ap{r}$. 
	In this simple example, the ``decision'' whether the formula will be satisfied (and which $\Gf{1}$ subformulas will be satisfied) comes after the first step. Thus, we can set $\Upsilon = \{(s_1,q_1)\}$.
\end{example}

\begin{wrapfigure}[10]{r}{0.35\textwidth}
	\begin{tikzpicture}
	\begin{scope}[x=0.4cm,y=0.4cm]
	\node[state,initial, initial below,initial text=] (a) {$q_0$}; 
	\node[state] (b) [right=1 of a] {$q_1$};
	\node[state] (c) [above=1 of b] {$q_2$};
	\node[state,accepting] (d) [left=1 of c] {$q_3$};
	\node[state] (e) [below=1 of b] {$q_4$};
	\path[-latex'] 
	(a) edge[] node[below=0.3]{$\true$} (b)
	(b) edge[] node[right=0.3]{$\neg \ap{q}$} (e)
	(b) edge[] node[right]{$\ap{q}$} (c)
	(c) edge[bend left] node[auto]{$\ap{m}$} (d)
	(d) edge[bend left] node[auto]{$\true$} (c)
	(c) edge[loop right] node[pos=0.7,below]{$\neg\ap{m}$} (c)
	(e) edge[loop right] node[pos=0.7,below]{$\true$} (e)
	;
	\end{scope}
	
	\begin{scope}[xshift=2.3cm,x=0.4cm,y=0.4cm]
	\node[state,initial, initial below,initial text=] (a) {$q_5$}; 
	\node[state] (b) [right=1 of a] {$q_6$};
	\node[state,accepting] (c) [above=1 of b] {$q_7$};
	\node[state] (d) [below=1 of b] {$q_8$};
	\path[-latex'] 
	(a) edge[] node[below, pos=0.5]{$\ap{q}$} (b)
	(a) edge[] node[above, pos=0.2]{$\neg\ap{q}$} (c)
	(b) edge[] node[auto,swap]{$\ap{r}$} (c)
	(b) edge[] node[auto]{$\neg \ap{r}$} (d)
	(c) edge[loop right] node[pos=0.3,above]{$\true$} (c)
	(d) edge[loop right] node[pos=0.7,below]{$\true$} (d)
	;
	\end{scope}
	\end{tikzpicture}
	\caption{Example Rabin aut.\label{fig:rabin}}
\end{wrapfigure}

We now prove Proposition~\ref{prop:reach}. The direction $\Leftarrow$ is straightforward. It suffices to define $\sigma$ so that it behaves as $\sigma'$ initially until it reaches some $(s,q) \in \Upsilon$ for the first time; then it behaves as $\sigma_{s,q}$.

Significantly more difficult is the direction $\Rightarrow$ of Proposition~\ref{prop:reach} that we address now. We fix a strategy $\sigma$ with $\Pr{\sigma}{\psi}=x$.
The proof is split into three steps. We first show how to identify the set $\Upsilon$, and then we show that items \ref{item:2a} and \ref{item:2b} of Proposition~\ref{prop:reach} are satisfied. The last part of the proof requires most of the effort.
In the proof, we will need to eliminate some unlikely events, and for this we will require that their probability is small to start with.
For this purpose, we fix a {\em very small} positive number \label{page:lambda}$\lambda$; to avoid cluttering of notation, we do not give
a precise value of $\lambda$, but instead point out that it needs to be chosen such that any numbers that depend on
it in the following text have the required properties (i.e. are sufficiently small or big; note that such choice is indeed possible). 
We should stress that $\lambda$ is influencing \emph{neither} the size of representation of our strategy \emph{nor} the complexity of our algorithm.
\paragraph*{Identifying the set $\Upsilon$}\label{page:sqI}

In the first step, we identify an appropriate set $\Upsilon$.
Intuitively, we put into $\Upsilon$ positions of the runs satisfying $\psi$ where \emph{the way $\psi$ is satisfied} is (nearly) decided, i.e. where it is (nearly) clear which frequency constraints will be satisfied by $\sigma$ in the future.
To this end, we mark every run $\omega$ satisfying $\psi$ with a set $I_\omega \subseteq \{1,\ldots,n\}$ such that $i\in I_\omega$ iff the formula
$\Gf{1} \varphi_i$ holds on the run. We then define a set of finite paths $\Gamma$ to contain all paths
$h$ for which there is $I_h \subseteq \{1,\ldots,n\}$ such that exactly the frequency constraints from $I_h$ as well as $\psi^{I_h}$ are satisfied on (nearly) all runs starting with $h$. Precisely, such that
$\Pr{\sigma}{\{\omega' \mid \omega' \models \psi^{I_h} \land I_{\omega'} = I_h\} \mid h} \ge 1-\lambda$.
Finally, for every $h \in \Gamma$ we add to $\Upsilon$ the pair $(s,q)$ where $h = h' s$ and $q$ is the state in which $\rabin_{\psi,I_h}$ ends after reading $\nu(h')$.

\paragraph*{Reaching $\Upsilon$}

It suffices to show that the strategy $\sigma$ itself satisfies $\Pr{\sigma}{\Gamma} = x$. We will use the following variant of L\'evy's Zero-One Law, a surprisingly powerful formalization of the intuitive fact that ``things need to get (nearly) decided, eventually''.

\begin{lemma}[L\'evy's Zero-One Law~\cite{chung}]\label{lemma:levy}
	Let $\sigma$ be a strategy and $X$ a measurable set of runs. Then for \emph{almost} every run $\omega$ we have
	$\lim_{n\to\infty} \Pr{\sigma}{X \mid h_n} = \idf_X(\omega)$
	where each $h_n$ denotes the prefix of $\omega$ with $n$ states and the function $\idf_X$ assigns $1$ to $\omega \in X$ and $0$ to $\omega \not\in X$.
\end{lemma}
\noindent
For every $I \subseteq \{1,\ldots,n\}$ we define $X_I = \{ \omega' \mid \omega' \models \psi^{I} \land I_{\omega'} = I \}$ to be the set of runs that are marked by $I$ and satisfy the formula $\psi^I$. 
Then by Lemma~\ref{lemma:levy}, for almost every run $\omega$ that satisfies $\psi$ and has $I_\omega = I$, there must be a prefix $h$ of the run for which $\Pr{\sigma}{ X_I \mid h} \geq 1-\lambda$ because $\omega \in X_I$. Any such prefix was added to $\Gamma$, with $I_h = I$.

\paragraph*{Almost-surely winning from $\Upsilon$}

For the third step of the proof of direction $\Rightarrow$ of Proposition~\ref{prop:reach} we fix $(\fixs,\fixq) \in \Upsilon$ and we construct a strategy $\sigma_{\fixs,\fixq}$ that is almost-surely winning from $(\fixs,\fixq)$. Furthermore, let $\fixI\subseteq\{1,\ldots,n\}$ denote the set such that $\fixq$ is a state from $\rabin_{\psi,\fixI}$. 
As we have shown in Theorem~\ref{thm:inf}, strategies might require infinite memory, and this needs to be taken into
consideration when constructing $\sigma_{\fixs,\fixq}$. The strategy cycles through two ``phases``, called {\em accumulating} and {\em reaching} that we illustrate on our example.

\begin{example}\label{ex:accum-intuition} 
	Returning to Example~\ref{ex:upsilon-intuition}, 
	we fix $(\fixs,\fixq) = (s_1,q_1)$ and $\fixI = \{1\}$, with the corresponding history from $\Gamma$ being $s_0 w s_1$. The strategy $\sigma_{s_1,q_1}$ we would like to obtain
	\begin{itemize} 
		\item first ``accumulates'' arbitrarily many steps from which all $\varphi_1^{\{1\}}$ can be almost surely satisfied. I.e., it accumulates arbitrarily many newly started instances of the Rabin automaton $\rabin_{\varphi_1,\{1\}}$ (all being in state $q_5$) by repeating action $w$ in $s_0$.
		\item Then it ``reaches'' with all the Rabin automata $\rabin_{\psi,\{1\}}$ and $\rabin_{\varphi_1,\{1\}}$ accumulated in the previous phase their accepting states $q_3$ and $q_7$ respectively. For $\rabin_{\varphi_1,\{1\}}$ this happens without any intervention of the strategy,
		but for $\rabin_{\psi,\{1\}}$ the strategy needs to take the action $m$. Then after returning to $s_0$ it comes back to a state where the next accumulating
		phase starts.
		Thus, we need to make sure we make the accumulating phases progressively longer so that in the long run they take place with frequency 1.
	\end{itemize}
	The proof that such a simple behaviour suffices is highly non-trivial.
	To illustrate this, let us extend the MDP from Figure~\ref{fig:intro} with an action \ap{decline} with $\Delta(s_1,\ap{decline})=s_0$.
	The strategy $\sigma$ from the proof of Theorem~\ref{thm:inf} satisfies $\Pr{\sigma}{\psi}=1$ for $\psi= \G\F \ap{m} \;\land\; \Gf{1} (\ap{q} \rightarrow \X \ap{r} )$.
	However, we can modify $\sigma$ and obtain a ``weird'' strategy $\sigma'$ that takes the action \ap{decline} in the $i$-th visit to $s_1$ with probability $1/2^i$. Such a strategy (a) still satisfies $\Pr{\sigma'}{\psi}=1/2$ but (b) it does not guarantee almost sure satisfaction of $\varphi_1^{\{1\}}$ in $s_1$. Thus, it does not accumulate in the sense explained above. We will show that any such weird strategy can be slightly altered to fit into our scheme. \qed
\end{example}

To show that alternation between such accumulating and reaching suffices (and to make a step towards the algorithm to construct such $\sigma_{\fixs,\fixq}$), we introduce a tailor-made product construction $\mdp_\otimes$. The product keeps track of a \emph{collection} of arbitrarily many Rabin automata accumulated up to now. 
We need to make sure that almost all runs of all automata in the collection are accepting, and we will do this by ensuring that:
(i) almost every computation of all Rabin automata eventually commits to an accepting condition $(E,F)$, and (ii) from the point the automaton
``commits'' to the accepting condition, no more elements of $E$ are visited and (iii) some element of $F$ is visited infinitely often.
To ensure this, we store additional information along any state $q \in Q$ of each automaton:
\begin{itemize}
	\item $(q,\symbopen)$ is a new instance that has to commit to an accepting condition soon;
	\item $(q,(E,F)_\circ)$ is an instance that has to visit a state of $F$ soon;
	\item $(q,(E,F)_\bullet)$ is an instance that recently fulfilled the accepting condition by visiting $F$;
	\item $(q,\bot)$ is an instance that violated the accepting condition it had committed to.
\end{itemize}

\noindent
Let $\mathcal{C}$ denote the set of these pairs for all $q \in Q$ and all accepting conditions $(E,F)$ of the Rabin automaton where the state $q$ belongs.
Note that $\mathcal{C}$ is finite; because we need to encode unbounded number of instances of Rabin automata along the run, each element of a collection $C \subseteq \mathcal{C}$ might stand for multiple instances that are in exactly the same configuration.
We say that $C\subseteq \mathcal{C}$ is {\em fulfilled} if it contains only elements of the form 
$(q,(E,F)_\bullet)$.
The aim is to fulfil the collection infinitely often, the precise meaning of ``recently'' and ``soon'' above is then ``since the last fulfilled state'' and ``before the next fulfilled state''. 

Using the product $\mdp_{\otimes}$, we show that if there is a satisfying strategy in $\mdp$, there is a  strategy in $\mdp_{\otimes}$ with a simple structure that visits a fulfilled state infinitely often (in Lemma~\ref{prop:screwing-paths}). Due to the simple structure, such a strategy can be found algorithmically. Finally, we show that such a strategy in the product induces a satisfying strategy in $\mdp$ (in Lemma~\ref{lemma:prod-to-orig}) yielding correctness of the algorithm.

\subparagraph{The product}
Let $\mdp_{\otimes}$ be an MDP with states $\states_\otimes = \states \times 2^{\mathcal{C}}$,
actions
$\actions_\otimes = \actions \times 2^{\mathcal{C}}$, and transition function $\tra_\otimes$ defined as follows.
We first define possible choices of a strategy in $\mdp_\otimes$. Given a state $(s,C_s)$, we say that an action $(a,C_a)$ is \emph{legal} in $(s,C_s)$ if
$a$ is a valid choice in $s$ in the original MDP, i.e. $\tra(s,a)$ is defined; and $C_a$ satisfies the following:
\begin{itemize}
	\item for all tuples $(q,\symbopen)\in C_s$ we have $(q,\symbopen)\in C_a$ or $(q,(E,F)_\circ)\in C_a$ for some accepting condition $(E,F)$,
	{\footnotesize ($q$ can ``commit'' to $(E,F)$, or keep waiting)}
	\item for all $(q,x)\in C_s$ with $x{\neq}\symbopen$ we have $(q,x)\in C_a$, {\footnotesize (all $q$ are kept along with the commitments)}
	\item all $(q,x) \in C_a$, not added by one of the two above items, are of the form $(q_{in}, \symbopen)$ where $q_{in}$ is the initial state of a Rabin automaton $R_{\varphi_i,\fixI}$ for $i\in \fixI$, {\footnotesize (initial states can be added)}
\end{itemize}

\noindent
\label{page:det-evolv}%
The randomness in $\mdp_\otimes$ comes only from $\mdp$. We set $\tra_\otimes((s,C_s),(a,C_a))(t,C_t) = \tra(s,a)(t)$
for any state $(s,C_s)$, any action $(a,C_a)$ legal in $(s,C_s)$, and any state $(t,C_t)$ such that $C_a$ ``deterministically evolves'' by reading $s$ into $C_t$. 
Precisely, we require that $C_t$ is the minimal set such that for any $(q,x) \in C_a$ there is $(q',x') \in C_t$ with $q \tran{\nu(s)} q'$ and $x \stackrel{q'}{\rightsquigarrow} x'$ where the latter relation is defined by $\symbopen \stackrel{q'}{\rightsquigarrow} \symbopen$ and $\bot \stackrel{q'}{\rightsquigarrow} \bot$  and for any $\cdot \in \{\bullet, \circ\}$ by
\begin{itemize}
	\item $(E,F)_\cdot \stackrel{q'}{\rightsquigarrow} (E,F)_\cdot$ if $q'\not\in E \cup F$ and $C$ is not fulfilled, {\footnotesize (no special state visited)}
	\item $(E,F)_\cdot \stackrel{q'}{\rightsquigarrow} (E,F)_\circ$ if $q'\not\in E \cup F$ and $C$ is fulfilled, {\footnotesize (resetting back to $\circ$)}
	\item $(E,F)_\cdot \stackrel{q'}{\rightsquigarrow} (E,F)_\bullet$  if $q' \in F$, {\footnotesize (the accepting condition becomes fulfilled)} 
	\item $(E,F)_\cdot \stackrel{q'}{\rightsquigarrow} \bot$ if $q' \in E$; {\footnotesize (the accepting condition is violated)} 
\end{itemize}

\noindent
Finally, a state is called {\em fulfilled} if its second component is fulfilled.

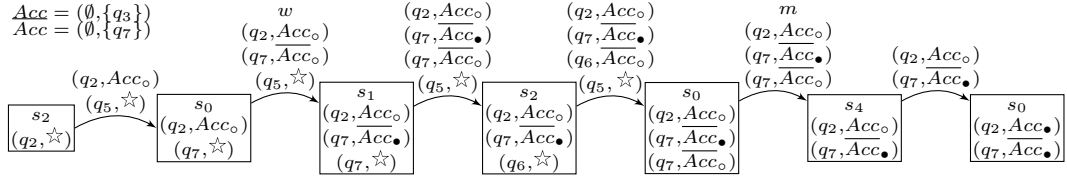
\begin{figure*}
	\begin{center}
		\begin{tikzpicture}[font=\scriptsize,inner ysep=0mm,inner xsep=-1.5mm,x=1.07cm]
		\node at (1.5,1.52) {$Acc=(\emptyset{,}\{q_3\})$};
		\node at (1.5,1.32) {$\overline{Acc}=(\emptyset{,}\{q_7\})$};
		\node[draw] (v1) at (1,0) {$\begin{array}{c}s_2\\(q_2{,}\symbopen)\end{array}$};
		\node[draw] (v2) at (3,0) {$\begin{array}{c}s_0\\(q_2{,}Acc_\circ)\\(q_7{,}\symbopen)\end{array}$};
		\node[draw] (v3) at (5,0) {$\begin{array}{c}s_1\\(q_2{,}Acc_\circ)\\(q_7{,}\overline{Acc}_\bullet)\\(q_7{,}\symbopen)\end{array}$};
		\node[draw] (v6) at (7,0) {$\begin{array}{c}s_2\\(q_2{,}Acc_\circ)\\(q_7{,}\overline{Acc}_\bullet)\\(q_6{,}\symbopen)\end{array}$};
		\node[draw] (v7) at (9,0) {$\begin{array}{c}s_0\\(q_2{,}Acc_\circ)\\(q_7{,}\overline{Acc}_\bullet)\\(q_7{,}\overline{Acc}_\circ)\end{array}$};
		\node[draw] (v8) at (11,0) {$\begin{array}{c}s_4\\(q_2{,}Acc_\circ)\\(q_7{,}\overline{Acc}_\bullet)\end{array}$};
		\node[draw] (v9) at (13,0) {$\begin{array}{c}s_0\\(q_2{,}Acc_\bullet)\\(q_7{,}\overline{Acc}_\bullet)\end{array}$};%
		
		\draw (v1.0) edge[bend left,-latex'] node[midway,above] {$\begin{array}{c}(q_2{,}Acc_\circ)\\(q_5{,}\symbopen)\end{array}$}  (v2.180);
		\draw (v2.28) edge[bend left,-latex'] node[midway,above] {$\begin{array}{c}w\\(q_2{,}Acc_\circ)\\(q_7{,}\overline{Acc}_\circ)\\(q_5{,}\symbopen)\end{array}$}  (v3.152);
		\draw (v3.25) edge[bend left,-latex'] node[midway,above] {$\begin{array}{c}(q_2{,}Acc_\circ)\\(q_7{,}\overline{Acc}_\bullet)\\(q_7{,}\overline{Acc}_\circ)\\(q_5{,}\symbopen)\end{array}$}  (v6.155);
		\draw (v6.25) edge[bend left, -latex'] node[midway,above] {$\begin{array}{c}(q_2{,}Acc_\circ)\\(q_7{,}\overline{Acc}_\bullet)\\(q_6{,}\overline{Acc}_\circ)\\(q_5,\symbopen)\end{array}$}  (v7.155);
		\draw (v7.30) edge[bend left,-latex'] node[xshift=2mm,midway,above] {$\begin{array}{c}m\\(q_2{,}Acc_\circ)\\(q_7{,}\overline{Acc}_\bullet)\\(q_7{,}\overline{Acc}_\circ)\end{array}$}  (v8.150);
		\draw (v8.30) edge[bend left,-latex'] node[midway,above] {$\begin{array}{c}(q_2{,}Acc_\circ)\\(q_7{,}\overline{Acc}_\bullet)\end{array}$}  (v9.150);
		\end{tikzpicture}
	\end{center}
	\caption{Illustration for Example~\ref{ex:product}. The names of actions from $\mdp$ are omitted when only a single action is available.\label{fig:product}}
\end{figure*}

\begin{example}\label{ex:product}
	Figure~\ref{fig:product} shows one path in the product $\mdp_\otimes$ for the MDP and the Rabin automata from Example~\ref{ex:upsilon-intuition}. The path shown illustrates how the initial states can be added non-deterministically (in the first three steps), and then reaches a fulfilled state.
\end{example}

A very useful property of the product is that any strategy that ensures visiting fulfilled states infinitely often yields a strategy in the original MDP such that the automata the strategy added almost surely accept. This is formalised in the following lemma.

\begin{lemma}\label{lemma:prod-to-orig}
	For a deterministic strategy $\pi$ in $\mdp_\otimes$ there is a strategy $\pi'$ in $\mdp$ that for any $h = (s_0,C_0) \cdots (a_n,D_n)(s_{n+1},C_{n+1})$ with
	$\Pr[h]{\pi}{\{\text{fulfilled state visited infinitely often}  \}} = 1$:
	\begin{itemize}
		\item $\Pr[s_0]{\pi'}{s_0\ldots a_n s_{n+1}} =\Pr[(s_0,C_0)]{\pi}{h}$, and 
		\item for any $(q,\symbopen) {\in} D_n$ where $R$ is the automaton of $q$, 
		$\Pr[s_0 a_0 \cdots s_n]{\pi'}{\{\omega \mid \text{$\rabin^q$ accepts $\omega$}\}} = 1$.
	\end{itemize}
\end{lemma}

To be able to use above lemma, we need to establish that it is {\em sufficient} to look for a strategy that visits fulfilled states infinitely often. In other words that existence of the satisfying strategy $\sigma$ implies existence of a strategy that visits fulfilled states infinitely often.
Here we use the following lemma saying that $\sigma$ and $(\fixs,\fixq)$ give rise to two
strategies in the product $\mdp_\otimes$ that can be used to add initial states into the collections, and to
reach fulfilled states. We will show below how these strategies can be used to finish the proof of Proposition~\ref{prop:reach}.

\begin{lemma}\label{prop:screwing-paths}
Assume $\fixs, \fixq, \fixI$ are chosen as described on page~\pageref{page:sqI}.	
	Then there are sets $M \subseteq S_\otimes$, $N\subseteq A_\otimes$ where 
	$N$ contains only ``accumulating`` actions, i.e. actions $(a,C)$ with $\{(q_{in}, \symbopen) \mid q_{in} \text{ is the initial state of $\rabin_{\varphi_i,\fixI}$ for } i \in \fixI\} \subseteq C$;
	and there are finite-memory deterministic strategies $\pi$ and $\zeta$ such that:
	\begin{enumerate} 
		\item\label{item:screwing-a}
		When starting in $(s,C) \in M$, $\pi$ only uses actions from $N$ and never leaves $M$
		\item\label{item:screwing-c} When starting in $(s,C) \in M\cup\{(\fixs,\{(\fixq,\symbopen)\})\}$,  $\zeta$ almost surely reaches a fulfilled state (possibly leaving $M$) and then reaches $M$.
	\end{enumerate}
\end{lemma}

\begin{proof}[Proof idea]
	The proof is involved and gives a crucial insight into the main obstacles of the proof of Theorem~\ref{thm:main-mdp}.
	Due to the space constraints we only sketch it here.
	
	We first prove that for any fixed $\ell$, almost every $\omega$ that satisfies all $\Gf{1} \varphi^\fixI_i$ has infinitely many {\em good} prefixes. Intuitively, a finite path $h$ is \emph{good} if, when starting from $h$,
	all the automata $\rabin_{\varphi_i,\fixI}$ for $i\in \fixI$ started within $\ell$ first steps accept with probability at least $1-\lambda$.
	
	In the second step, we show how to avoid actions that cause that any $\rabin_{\varphi_i,\fixI}$ does not accept.
	To do so, we inductively start labelling the prefixes of runs of the MDP with elements of $\mathcal{C}$. Having fixed a label for a prefix, the label for its extension is obtained by ``deterministic evolving'' as in the definition of the product MDP, and by (non-deterministically) adding $(q_{in},\symbopen)$.
	The latter part is performed by switching between a ``pseudo-accumulating'' and ``pseudo-reaching'' phase. Initially, we start in a pseudo-reaching phase, only with singletons corresponding to the current state of $\rabin_{\psi,\fixI}$, and do not add any $(q_{in},\symbopen)$. When a good prefix is reached (which happens almost surely), we switch to a pseudo-accumulating phase for the next $\ell$ steps and we keep adding ``initial states'' $(q_{in},\symbopen)$ of $\rabin_{\varphi_i,\fixI}$ for each $i\in\fixI$. After $\ell$ steps, we switch back to a pseudo-reaching phase and do not add any new elements to the label until we pass through a state whose label is fulfilled and get to a good prefix again, in which point another pseudo-accumulating phase starts.
	
	Along the way, we might obtain tuples of the form $(q,\bot)$ in the label, or we might not ever visit a fulfilled state. Indeed, if we repeated our steps to infinity, such an ``error'' might take place almost surely. However, before an error happens with too high probability, the labels start repeating because $\mathcal{C}$ is finite. 
	We show that supposing $\ell$ was large enough and our tolerance $\lambda$ was small enough, there must be a strategy that \emph{almost-surely} traverses such a cycle without any error.
	We can extract from the pseudo-accumulating and pseudo-reaching phases of such a strategy the sets $M$ (and $N$), given by the tuples of the MDP states (actions) and their labels.
\end{proof}

We are now ready to finish the proof of Proposition~\ref{prop:reach}. We show that Lemma~\ref{prop:screwing-paths} allows us to
construct a strategy $\sigma_\otimes$ for $\mdp_\otimes$ that almost surely (i) visits fulfilled states, and (ii) with frequency $1$ it takes actions from $N$. By Lemma~\ref{lemma:prod-to-orig} this strategy yields an almost-surely winning strategy $\sigma_{\fixs,\fixq}$ in $\mdp$.

The strategy $\sigma_\otimes$ is constructed as follows.
Inductively, for path $h$ in $\mdp_\otimes$, we say that its first accumulating phase starts in the first step, $i$th accumulating phase takes $i$ steps, and the $(i+1)$th accumulating phase starts when the set $M$ is reached through a fulfilled state after the $i$th accumulating phase ended.
Within every accumulating phase started in a history $h$, $\sigma_\otimes$ is defined to play as $\pi$ initiated after $h$. Similarly, outside every accumulating phase ended in a history $h$, $\sigma_\otimes$ is defined to play as $\zeta$.

\subsection{The algorithm}
\label{sec:alg}
To conclude the proof of Theorem~\ref{thm:main-mdp}, we need to give a procedure for  
computing the optimal probability of satisfying $\psi$.
It works in the following steps (for details, see~\citeapp{app:algorithm}):
\begin{enumerate}
	\item\label{item:alg-a} Initialize $\Upsilon := \emptyset$, and construct $\rabin_{\xi,I}$ for all $\xi\in \{\psi,\varphi_1,\ldots,\varphi_n\}$ and $I\subseteq \{1,\ldots,n\}$.
	\item\label{item:alg-b} For every $I$ 
	find the largest sets $(M_I,N_I)$ satisfying the conditions \ref{item:screwing-a}--\ref{item:screwing-c} of Lemma~\ref{prop:screwing-paths}, and add to $\Upsilon$ all pairs $(s,q)$ such that $M_I$ can be almost-surely reached from 
	$(s,\{(q,\symbopen)\})$.
	\item\label{item:reach} Compute an optimal strategy $\sigma'$ for ``reaching'' $\Upsilon$  and return the probability. Intuitively,
	\begin{itemize}
		\item we build the ``naive'' product of $\mdp$ with all the main automata $\rabin_{\psi,I}$ for $I\subseteq \{1,\ldots,n\}$;
		\item 
		reaching $\Upsilon$ is reduced to ordinary reachability of all states of the form $(s,q_1,\ldots,q_m)$ such that $(s,q_i)\in \Upsilon$ for some $i$.
		\item
		By standard algorithms for reachability in MDP, we find an optimal strategy $\sigma''$ in the naive product that easily induces the strategy $\sigma'$ in $\mdp$.
	\end{itemize}
\end{enumerate}

\noindent
By connecting Proposition~\ref{prop:reach}, Lemmas~\ref{lemma:prod-to-orig} and~\ref{prop:screwing-paths}, and the construction of $\sigma_\otimes$ above, there is a strategy $\sigma$ in $\mdp$ yielding probability $\geq p$ iff the set $\Upsilon$ computed by the algorithm can also be reached with probability $\geq p$.

We briefly discuss the complexity of the algorithm. Each of the Rabin automata in step \ref{item:alg-a} above can be computed in time $2^{2^{\poly{|\varphi|}}}$,
and since there is exponentially many such automata (in $|\varphi|$), step 1. takes time $2^{2^{poly(|\varphi|)}}$.
Step \ref{item:alg-b} can be performed in time $\poly{S}\cdot 2^{2^{\poly{|\varphi|}}}$.
In step \ref{item:reach} we are computing reachability probability in the naive product MDP which is of size $\poly{S} \cdot 2^{2^{\poly{|\varphi|}}}$,
and so also this step can be done in time $\poly{S}\cdot 2^{2^{\poly{|\varphi|}}}$.

\section{Conclusions}

We have given algorithms for controller synthesis of the logic LTL extended with
an operator expressing that frequencies of some events exceed a given bound.
For Markov chains we gave an algorithm working with the complete logic, and for
MDPs we require the formula to be from a certain fragment.
The obvious next step is extending the MDP results to the whole fLTL.
This will require new insights. Our product construction relies on the (non-trivial) observation
that given $\Gf{1}\varphi$, the formula $\varphi$ is almost surely satisfied from
any history of an accumulating phase. This is no longer true when the frequency bound is
lower than $1$. In such cases different histories may require different probability of
satisfying $\varphi$. However, both authors strongly believe that even for these cases the problem
is decidable.
Another promising direction for future work is implementing the algorithms into a probabilistic model checker and evaluating their time requirements experimentally.

\subparagraph{Acknowledgements.} The authors would like to thank anonymous reviewers for their insightful comments on an earlier version of this paper. 
This work is partly supported by the German Research Council (DFG) as part of the Transregional Collaborative Research Center AVACS (SFB/TR 14), by the Czech Science Foundation under grant agreement P202/12/G061, by the EU 7th Framework Programme under grant agreement no. 295261 (MEALS) and 318490 (SENSATION), by the CAS/SAFEA International Partnership Program for Creative Research Teams, and EPSRC grant EP/M023656/1. Vojt\v{e}ch Forejt is also affiliated with Faculty of Informatics, Masaryk University, Brno, Czech Republic.

\bibliographystyle{plain}
\bibliography{bib}

\begin{thebibliography}{10}

\bibitem{BHHK00}
Christel Baier, Boudewijn Haverkort, Holger Hermanns, and Joost-Pieter Katoen.
\newblock Model checking continuous-time {M}arkov chains by transient analysis.
\newblock In {\em CAV}, volume 1855 of {\em LNCS}. Springer, 2000.

\bibitem{BP08}
Christel Baier and Joost-Pieter Katoen.
\newblock {\em {Principles of model checking}}.
\newblock MIT Press, 2008.

\bibitem{baier2014weight}
Christel Baier, Joachim Klein, Sascha Kl{\"u}ppelholz, and Sascha Wunderlich.
\newblock Weight monitoring with linear temporal logic: Complexity and
  decidability.
\newblock In {\em CSL-LICS}, page~11. ACM, 2014.

\bibitem{bloem2009better}
Roderick Bloem, Krishnendu Chatterjee, Thomas~A Henzinger, and Barbara
  Jobstmann.
\newblock Better quality in synthesis through quantitative objectives.
\newblock In {\em CAV}, pages 140--156. Springer, 2009.

\bibitem{boker2011temporal}
Udi Boker, Krishnendu Chatterjee, Thomas~A Henzinger, and Orna Kupferman.
\newblock Temporal specifications with accumulative values.
\newblock In {\em LICS 2011}, pages 43--52. IEEE, 2011.

\bibitem{BDL-tase12}
Benedikt Bollig, Normann Decker, and Martin Leucker.
\newblock Frequency linear-time temporal logic.
\newblock In {\em {TASE}'12}, pages 85--92, Beijing, China, July 2012. {IEEE}
  Computer Society Press.

\bibitem{BMM14}
Patricia Bouyer, Nicolas Markey, and Raj~Mohan Matteplackel.
\newblock Averaging in {LTL}.
\newblock In Paolo Baldan and Daniele Gorla, editors, {\em {CONCUR} 2014},
  volume 8704 of {\em LNCS}, pages 266--280. Springer, 2014.

\bibitem{BFK08}
Tom{\'a}\v{s} Br{\'a}zdil, Vojt\v{e}ch Forejt, and Anton\'{\i}n Ku\v{c}era.
\newblock Controller synthesis and verification for markov decision processes
  with qualitative branching time objectives.
\newblock In {\em ICALP 2008}, volume 5126 of {\em LNCS}, pages 148--159.
  Springer, 2008.

\bibitem{chatterjee2011energy}
Krishnendu Chatterjee and Laurent Doyen.
\newblock Energy and mean-payoff parity {M}arkov decision processes.
\newblock In {\em MFCS 2011}, pages 206--218. Springer, 2011.

\bibitem{chatterjee2012games}
Krishnendu Chatterjee and Laurent Doyen.
\newblock Games and markov decision processes with mean-payoff parity and
  energy parity objectives.
\newblock In {\em MEMICS}, pages 37--46. Springer, 2012.

\bibitem{chatterjee2005mean}
Krishnendu Chatterjee, Thomas~A Henzinger, and Marcin Jurdzinski.
\newblock Mean-payoff parity games.
\newblock In {\em LICS 2005}, pages 178--187. IEEE, 2005.

\bibitem{chung}
Kai-lai Chung.
\newblock {\em A Course in Probability Theory}.
\newblock Academic Press, 3 edition, 2001.

\bibitem{de1998specify}
Luca De~Alfaro.
\newblock How to specify and verify the long-run average behaviour of
  probabilistic systems.
\newblock In {\em LICS 1998}, pages 454--465. IEEE, 1998.

\bibitem{KSK76}
J.~Kemeny, J.~Snell, and A.~Knapp.
\newblock {\em Denumerable {M}arkov Chains}.
\newblock Springer, 2nd edition, 1976.

\bibitem{kuvcera2005controller}
Anton{\'\i}n Ku{\v{c}}era and Old{\v{r}}ich Stra{\v{z}}ovsk{\'y}.
\newblock On the controller synthesis for finite-state {M}arkov decision
  processes.
\newblock In {\em FSTTCS 2005}, pages 541--552. Springer, 2005.

\bibitem{KNP12a}
M.~Kwiatkowska, G.~Norman, and D.~Parker.
\newblock Probabilistic verification of herman’s self-stabilisation
  algorithm.
\newblock {\em Formal Aspects of Computing}, 24(4):661--670, 2012.

\bibitem{lyons}
Russell Lyons.
\newblock Strong laws of large numbers for weakly correlated random variables.
\newblock {\em Michigan Math. J.}, 35(3):353--359, 1988.

\bibitem{Norris}
J.~R. Norris.
\newblock {\em Markov chains}.
\newblock Cambridge University Press, 1998.

\bibitem{Shm04}
V.~Shmatikov.
\newblock Probabilistic model checking of an anonymity system.
\newblock {\em Journal of Computer Security}, 12(3/4):355--377, 2004.

\bibitem{DBLP:journals/corr/abs-1111-3111}
Takashi Tomita, Shigeki Hagihara, and Naoki Yonezaki.
\newblock A probabilistic temporal logic with frequency operators and its model
  checking.
\newblock In {\em INFINITY}, volume~73 of {\em {EPTCS}}, pages 79--93, 2011.

\bibitem{AT12}
Takashi Tomita, Shin Hiura, Shigeki Hagihara, and Naoki Yonezaki.
\newblock A temporal logic with mean-payoff constraints.
\newblock In {\em Formal Methods and Soft. Eng.}, volume 7635 of {\em LNCS}.
  Springer, 2012.

\end{thebibliography}

\onecolumn
\newpage
\appendix
\section{Details for proof for Markov chains}

\begin{reflemma}{lemma:alt-slln}
Let $Y_1,Y_2\ldots$ be a sequence of random variables which only take values $0$ or $1$ and have expectation $\mu$. Assume there are $0<r,c<1$ such that for all $i,j\in \Nset$ we have $\E{(Y_i-\mu)(Y_j - \mu)} \le r^{\lfloor c |i-j|\rfloor}$. Then
$\lim_{n\rightarrow \infty} \sum_{i=0}^n Y_i / n = \mu$ almost surely.
\end{reflemma}
\begin{proof}
We can use \cite[Corollary 4]{lyons} applied to random variables $Z_i = Y_i - \mu$ (we cannot use the result directly for $Y_i$ since \cite{lyons} requires the random variables to have expectation value equal to $0$). Clearly if $\lim_{n\rightarrow \infty} \sum_{i=0}^n Z_i / n = 0$,
then
$\lim_{n\rightarrow \infty} \sum_{i=0}^n Y_i / n = \lim_{n\rightarrow \infty} \sum_{i=0}^n (Z_i + \mu) / n = \mu$. 
Finally, the corollary of \cite{lyons} indeed applies since
$\sum_{k=0}^\infty \frac{r^{\lfloor c \cdot k\rfloor}}{k} \le
 \sum_{k=0}^\infty r^{c \cdot k - 1} = 
 r^{-1}\cdot\sum_{k=0}^\infty (r^c)^{k} \le 
 r^{-1}\cdot\sum_{k=0}^\infty (r^c)^{k} \le 
 1/r(1-r^c)
 < \infty$
\end{proof}
\subsection{Properties of random variables $X^t_i$}
\label{app:rv-xti}
The following is a more detailed computation for properties of the random variables $X^t_i$. First, we need to extend the definition of a path being determined. We say that a path $h$ is {\em positively determined} if almost every run of $\cyl(s_0\ldots s_k)$ satisfies $\varphi$, and {\em negatively determined} if almost no run of $\cyl(s_0\ldots s_k)$ satisfies $\varphi$. Now splitting runs of $D$ to $D^+$ and $D^-$ depending on whether the associated path
is positively or negatively determined, we have:
\begin{align*}
 \E[s]{&(X^t_i - \mu_t) (X^t_j - \mu_t)} \\
& = \P[s]{D^+}\cdot \E[s]{(1-\mu_t) (X^t_j - \mu_t) \mid D^+} \\
& \quad + \P[s]{D^-} \cdot \E[s]{(-\mu_t) (X^t_j - \mu_t) \mid D^-}\\
& \quad + \P[s]{\Omega\setminus(D^+\cup D^-)}\cdot\E[s]{(X^t_i - \mu_t) (X^t_j - \mu_t) \mid \Omega\setminus(D^+\cup D^-)}\\ 
& = \P[s]{D^+}\cdot (1-\mu_t)\cdot \E[s]{X^t_j - \mu_t \mid D^+}\\
& \quad+ \P[s]{D^-} \cdot (-\mu_t)\cdot \E[s]{X^t_j - \mu_t \mid D^-} \\
& \quad + \P[s]{\Omega\setminus(D^+\cup D^-)}\cdot\E[s]{(X^t_i - \mu_t) (X^t_j - \mu_t) \mid \Omega\setminus(D^+\cup D^-)}\\ 
\intertext{and because $\E[s]{X^t_j - \mu_t \mid D^+} = \E[s]{X^t_j - \mu_t) \mid D^-} = 0$ (as shown later), we have}
& = \P[s]{\Omega\setminus(D^+\cup D^-)}\cdot\E[s]{(X^t_i - \mu_t) (X^t_j - \mu_t) \mid \Omega\setminus(D^+\cup D^-)}\\ 
& \le (1-\P[s]{D^+\cup D^-}) \cdot 1\\ 
& \le (1-p^M)^{\lfloor (i-j)/M\rfloor}
\end{align*}

Now let us show that $\E[s]{X^t_j - \mu_t \mid D^+} = 0$ by showing that $\E[s]{X^t_j\mid D^+}=\mu_t$. The argument for $D^-$ is analogous.

Let $h_1,h_2\ldots$ be the sequence of all finite paths ending in $t$, containing $j$ occurrences of $t$, and satisfying the condition that
the suffix from $i$th to $j$th occurrence of $t$ is positively determined. The sets $h_k\cap D^+$ partition $D^+$ and moreover
$\P[s]{h_k} = \P[s]{h_k\cap D^+}$. Hence, we have

\begin{multline*}
\E[s]{X^t_j \mid D^+} = \sum_{k} \P[s]{h_k\mid D^+}\cdot \E[s]{X^t_j \mid h_k} = \sum_{k} \P[s]{h_k\mid D^+}\cdot \E[s]{X^t_j}\\= \mu_t\cdot \sum_{k} \P[s]{h_k\mid D^+} = \mu_t
\end{multline*}
where the second equality follows because the value of $X^t_j$ does not depend on the prefix up to the $j$th visit to $t$, and because $h_k$ is a cylindric set.
\subsection{Relating frequency with probability of achieving bsccs.}
\label{app:mc-bscc}
The following is the final computation for the proof of Proposition~\ref{prop:mc-replace}. Below, we use $N^t_n$ for the random variable
that for a run $\omega$ returns the number of visits to $t$ on the prefix of $\omega$ of length $n$.
\begin{align*}
\liminf_{n\rightarrow\infty} \sum_{i=0}^n 1_{\varphi,i} / n &= \liminf_{n\rightarrow\infty} \sum_{t \in\bscc{s}} \sum_{i=0}^{N_n^t} X^t_i / n\\
&= \sum_{t \in\bscc{s}} \liminf_{n\rightarrow\infty} \sum_{i=0}^{N_n^t} X^t_i / n\\
&= \sum_{t \in\bscc{s}} \liminf_{n\rightarrow\infty} \sum_{i=0}^{n} X^t_i / (n/\ssf{t})\\
&= \sum_{t \in\bscc{s}} \ssf{t} \liminf_{n\rightarrow\infty} \sum_{i=0}^{n} X^t_i / n\\
&= \sum_{t \in\bscc{s}} \ssf{t} \P[s]{\varphi}\\
\end{align*}

\section{Details for proof for Markov decision processes} \label{app:algorithm}

\subsection{Proof of Lemma~\ref{lemma:prod-to-orig}}

\begin{reflemma}{lemma:prod-to-orig}
For a deterministic strategy $\pi$ in $\mdp_\otimes$ there is a strategy $\pi'$ in $\mdp$ that for any $h = (s_0,C_0) \cdots (a_n,D_n)(s_{n+1},C_{n+1})$ with
$\Pr[h]{\pi}{\{\text{fullfilled state visited infinitely often}  \}} = 1$:
\begin{itemize}
	\item $\Pr[s_0]{\pi'}{s_0\ldots a_n s_{n+1}} =\Pr[(s_0,C_0)]{\pi}{h}$, and 
	\item for any $(q,\symbopen) {\in} D_n$ where $R$ is the automaton of $q$, 
	$\Pr[s_0 a_0 \cdots s_n]{\pi'}{\{\omega \mid \text{$\rabin^q$ accepts $\omega$}\}} = 1$.
\end{itemize}
\end{reflemma}
\begin{proof}
	For every finite path $\hat h = s_0a_0 \cdots a_{n}s_{n+1}$ in $\mdp$ there is at most one path of the form $(s_0,C_0)(a_0,D_0) \cdots (a_{n},D_{n})(s_{n+1},C_{n+1}))$, denoted $\hat h_\otimes$ which satisfies that:
 \begin{itemize}
  \item  with $C_0$ fixed above
  \item  all $D_i$ are as chosen with probability $1$ by the deterministic strategy $\sigma_\otimes$ and
  \item  all $C_i$ all given uniquely by the definition of $\mdp_\otimes$
 \end{itemize}
	
	We define the strategy $\pi'$ by 
	$\pi'(\hat{h}) = \pi(\hat{h}_\otimes)$ for all $\hat h$ starting with $s_0$, and define $\pi'(h)$ arbitrarily otherwise.
       
	Let $h = (s_0,C_0) \cdots (a_n,D_n)(s_{n+1},C_{n+1})$ be a path. Clearly,
        $\Pr[s_0]{\pi'}{s_0\ldots a_n s_{n+1}} =\Pr[(s_0,C_0)]{\pi}{h}$ by the definition of $\pi$. Also,
	since $\pi$, when starting after $h$, almost surely fulfils infinitely often, it also never reaches any state with
        second component containing $(q,\bot)$. Hence, it is easy to see from the definition of
        $\mdp_\otimes$ that $\Pr[(s_0,C_0)]{\sigma_\otimes}{h}>0$, then 
	From the definition of $\mdp_\otimes$ it is easy to see that fulfilling infinitely often implies that for all $(q,\symbopen)\in D_n$ the automaton
        $\rabin^q$ (where $\rabin$ is the automaton containing $q$) almost surely accepts suffixes of $s_0a_0\ldots s_k$.
\end{proof}

\subsection{Proof of Lemma~\ref{prop:screwing-paths}}
\label{sec:prop-proof}

First of all, let us introduce further definitions that we will require later in the proof. For a run $\omega$ and a state $q$ of a Rabin automaton $\rabin$, by $\infi_q(\omega)$ (resp. $\occu_q(\omega)$) we denote the set of states of $\rabin$
that occur infinitely many times in $\rabin[\omega]$ (resp. that occur at least once in $\rabin[\omega]$)
when the Rabin automaton is started from state $q$ instead of $\rinit$.

\begin{reflemma}{prop:screwing-paths}
Assume $\fixs, \fixq, \fixI$ are chosen as described on page~\pageref{page:sqI}.	
	Then there are sets $M \subseteq S_\otimes$, $N\subseteq A_\otimes$ where 
	$N$ contains only ``accumulating`` actions, i.e. actions $(a,C)$ with $\{(q_{in}, \symbopen) \mid q_{in} \text{ is the initial state of $\rabin_{\varphi_i,\fixI}$ for } i \in \fixI\} \subseteq C$;
	and there are finite-memory deterministic strategies $\pi$ and $\zeta$ such that:
	\begin{enumerate} 
		\item\label{item:screwing-a}
		When starting in $(s,C) \in M$, $\pi$ only uses actions from $N$ and never leaves $M$
		\item\label{item:screwing-c} When starting in $(s,C) \in M\cup\{(\fixs,\{(\fixq,\symbopen)\})\}$,  $\zeta$ almost surely reaches a fulfilled state (possibly leaving $M$) and then reaches $M$.
	\end{enumerate}
\end{reflemma}

We will now prove Lemma~\ref{prop:screwing-paths}. As before, fix $\fixs\in S$, $\fixq\in Q$, and $\fixI\subseteq\{1,\ldots n\}$, and we also fix a finite path $\fixh$ from $\Gamma$ witnessing that $(\fixs,\fixq)\in \Upsilon$.
We also fix \label{page:ell}$\ell = |S_\otimes| + 2$ and $\kappa = 3 \cdot |\ell|^2 \cdot 2^n \cdot \lambda$ where $\lambda$ is the small number introduced at page~\pageref{page:lambda}.

The following definition and Lemma~\ref{lemma:brew-often} will help us identify (possible) recurring behaviour of $\sigma$. We need to identify long enough parts of runs where all the frequency formulae $\varphi_{i}^{\fixI}$ are satisfied with probability very close to $1$. Based on the behaviour of $\sigma$ within these parts, we later define the ``accumulating'' strategy.

For a path $h$, let $|h|$ be length of $h$, i.e. the number of states in $h$.
We say that a finite path $h$ extending $\fixh$ is {\em good} 
if 
\[
\Ex{\sigma}{\sum_{k=0}^{\ell-1} Y_{|h|+k} \mid h} > \ell\cdot(1-\lambda).
\]
where $Y_j(\omega) = \idf_{\forall i\in\fixI: \omega(j\ldots) \models \varphi_{i}^{\fixI}}$ is the indicator function that the suffix of $\omega$ starting at $j$-th position satisfies all $\varphi_{i}^{\fixI}$.

\begin{lemma}\label{lemma:prefix-approximate}
Let $X$ be a set of runs, $\beta > 0$, and let $J_i = \{h \mid \text{$|h|=i$ and $\mathbb{P}(X \mid h) \ge \beta$}\}$ then
$\lim_{i\to\infty} J_i = X$, i.e. for every $\omega$ there is $i$ such that
for all $i' > i$ we have $\omega\in J_i$ iff $\omega\in X$.
\end{lemma}
\begin{proof}
If $\omega\in X$, then by Lemma~\ref{lemma:levy} there is $i$ such that for all $i'>i$ we have
$\mathbb{P}(X\mid h) \ge \beta$ where $h$ is the prefix of $\omega$ of length $i'$. Then $i$ is the required number.
If $\omega\not\in X$, then again by Lemma~\ref{lemma:levy} there is $i$ such that for all $i'>i$ we have
$\mathbb{P}(X\mid h) < \beta$ where $h$ is the prefix of $\omega$ of length $i'$. Then again we pick $i$.
\end{proof}

The following lemma allows us to simplify the notation and only deal with one frequency-globally formula $\varphi := \bigwedge_{i\in\fixI} \varphi_{i,\fixI}$.
\begin{lemma}\label{lemma:one-formula}
	Let $\xi_1,\ldots \xi_n$ be LTL formulae, and $\omega$ a run. We have
	$\omega\models \bigwedge_{i=1}^n \Gf{1} \xi_i$
	if and only if
	$\omega\models \Gf{1}\bigwedge_{i=1}^n  \xi_i$.
\end{lemma}

\begin{lemma}\label{lemma:brew-often}
	Almost every $\omega$ satisfying $\bigwedge_{i\in \fixI} \Gf{1} \varphi^\fixI_i$
	has infinitely many good prefixes.
\end{lemma}
\begin{proof}
	By contradiction.
        Employing Lemma~\ref{lemma:one-formula} we can slightly simplify the problem and
        consider runs satisfying  $\Gf{=1}\varphi$ for $\varphi\equiv \bigwedge_{i\in \fixI} \varphi^\fixI_i$
        Suppose that there is a set $X'$ with $\Pr{\sigma}{X'} > 0$ such that
	all $\omega\in X'$ satisfy $\Gf{=1}\varphi$ and have only finitely
	many good prefixes.
	
	Further, let $m_\omega$ for a run $\omega\in X'$ denote the smallest number such that
	for all $m' \ge m_\omega$ the prefix $h$ of $\omega$ of length $m'$ is not good.
	We can pick $m$ and $X\subseteq X'$ satisfying that $\Pr{\sigma}{X} > 0$,
	and every $\omega\in X$ satisfies that $m_\omega \le m$. 
	Note that such choice is possible, as with increasing $m$ the set $X$ tends monotonically to $X'$.
	
	Note that we have
	\begin{align} \label{eq:freq}
	\Ex{\sigma}{\liminf_{n\to\infty} \frac{1}{n}\sum_{i=0}^n Y_i \mid X } = 1
	\end{align}
	
	Furthermore, by Fatou's Lemma, by linearity of expectation, and by taking a subsequence of averages of chunks of length $\ell$, we have
	\begin{align*}
	\Ex{\sigma}{&\liminf_{n\to\infty} \frac{1}{n}\sum_{i=0}^n Y_i\mid X} \\
	& \; \leq \;
	\liminf_{n\rightarrow \infty} \frac{1}{n} \sum_{i=0}^n
	\Ex{\sigma}{Y_i\mid X} \\
	& \; \leq \;
	\liminf_{k\rightarrow \infty} \frac{1}{k} \sum_{j=0}^{k-1} \frac{1}{\ell} \sum_{i=0}^{\ell-1}
	\Ex{\sigma}{Y_{\ell\cdot j + i} \mid X}
	\; = \; (\ast)
	\end{align*}
	
	Let $J_i$ be the set defined as in Lemma~\ref{lemma:prefix-approximate}, and denote $\compl{J_i} = \{h \mid |h| = i\} \setminus J_i$.
	For $\ell\cdot j \ge m$ we have
	\begin{align*}
	\lefteqn{\frac{1}{\ell} \sum_{i=0}^{\ell-1} \Ex{\sigma}{Y_{\ell\cdot j + i} \mid X}}\\
	& \; = \;
	\frac{1}{\ell} \sum_{|h| = \ell\cdot j} \Pr\sigma{h \mid X} \sum_{i=0}^{\ell-1} \Ex{\sigma}{Y_{\ell\cdot j + i}\mid h\cap X}\tag{law of total expectation}\\
	& \; = \;
	\frac{1}{\ell}\sum_{|h| = \ell\cdot j} \Pr\sigma{h \mid X} \sum_{i=0}^{\ell-1} \sum_{x\in \{0,1\}} x \cdot \Pr{\sigma}{Y_{\ell\cdot j + i} = 1 \mid h \cap X}\tag{def. of expectation}\\
	& \; = \;
	\frac{1}{\ell}\sum_{|h| = \ell\cdot j} \frac{\Pr\sigma{h \cap X}}{\Pr\sigma{X}} \sum_{i=0}^{\ell-1}  \frac{\Pr{\sigma}{Y_{\ell\cdot j + i} = 1 \cap h \cap X}}{\Pr{\sigma}{X\cap h}}\tag{removing $0$ terms; def. of cond. probability}\\
	& \; = \;
	\frac{1}{\ell} \sum_{|h| = \ell\cdot j} \frac{1}{\Pr\sigma{X}} \sum_{i=0}^{\ell-1} \Pr{\sigma}{Y_{\ell\cdot j + i} = 1 \cap h \cap X}\tag{$\Pr\sigma{h \cap X}$ cancels out}\\
	& \; = \;
	\frac{1}{\ell}  \frac{1}{\Pr\sigma{X}}\bigg( \sum_{h \in J_{\ell\cdot j}}\sum_{i=0}^{\ell-1} \Big(\Pr{\sigma}{Y_{\ell\cdot j + i} = 1 \cap h} - \Pr{\sigma}{Y_{\ell\cdot j + i} = 1 \cap (h\setminus X)}\Big)\\
	& \qquad \qquad + \sum_{h \in \compl{J_{\ell\cdot j}}}\sum_{i=0}^{\ell-1} \Pr{\sigma}{Y_{\ell\cdot j + i} = 1 \cap h \cap X}\bigg)\tag{partitioning in paths (not) in $J_{\ell\cdot j}$; set operation}\\
	& \; \le \;
	\frac{1}{\ell} \frac{1}{\Pr\sigma{X}} \bigg( \sum_{h \in J_{\ell\cdot j}}\sum_{i=0}^{\ell-1} \Pr{\sigma}{Y_{\ell\cdot j + i} = 1 \cap h} + \sum_{h \in \compl{J_{\ell\cdot j}}}\sum_{i=0}^{\ell-1} \Pr{\sigma}{h \cap X}\bigg)
       \tag{removing negative terms}\\
	& \; = \;
	\frac{1}{\ell}\frac{1}{\Pr\sigma{X}}\bigg(\sum_{h \in J_{\ell\cdot j}} \sum_{i=0}^{\ell-1}  \Pr\sigma{h} \frac{\Pr{\sigma}{Y_{\ell\cdot j + i} = 1 \cap h}}{\Pr\sigma{h}} + \sum_{h \in \compl{J_{\ell\cdot j}}}\ell\cdot  \Pr{\sigma}{h \cap X}\bigg)\tag{multiplying some summands by $\Pr\sigma{h}/\Pr\sigma{h}$}\\
	& \; = \;
	\frac{1}{\ell} \frac{1}{\Pr\sigma{X}}\bigg(\sum_{h \in J_{\ell\cdot j}} \Pr\sigma{h} \sum_{i=0}^{\ell-1}\Pr{\sigma}{Y_{\ell\cdot j + i} = 1 \mid h} + \ell \cdot \Pr{\sigma}{\compl{J_{\ell\cdot j}} \cap X}
	\tag{def. of cond. prob.}\\
	& \; \le \;
	\frac{1}{\ell} \frac{1}{\Pr\sigma{X}}\bigg( \Pr\sigma{J_{\ell\cdot j}} \cdot \ell\cdot(1-\lambda) + \ell\cdot \Pr{\sigma}{\compl{J_{\ell\cdot j}} \cap X}\bigg)\tag{property of $J_{\ell\cdot j}$ for $\ell\cdot j\ge m$}\\
	& \; \le \;
	\frac{\Pr\sigma{J_{\ell\cdot j}} \cdot (1-\lambda)}{\Pr\sigma{X}} + \frac{\Pr{\sigma}{\compl{J_{\ell\cdot j}} \cap X}}{\Pr\sigma{X}}
	\end{align*}
	and hence
	\begin{align*}
	(\ast) &\le \liminf_{k\rightarrow \infty} \frac{1}{k} \sum_{j=0}^{k-1} \Big(\frac{\Pr\sigma{J_{\ell\cdot j}} \cdot (1-\lambda)}{\Pr\sigma{X}} + \frac{\Pr{\sigma}{\compl{J_{\ell\cdot j}} \cap X}}{\Pr\sigma{X}}\Big)\\
	\intertext{and since in $\lim_{\ell\cdot j \rightarrow \infty} J_{\ell\cdot j} = X$, we get}
	&= (1-\lambda)
	\end{align*}
	which is a contradiction with (\ref{eq:freq}).
\end{proof}

By heavily relying on existence of good prefixes, we define labellings of histories of $\mdp$ that will help us establish a connection to $\mdp_\otimes$. Namely, the labellings (1) identify what is the current state in $\mdp_\otimes$ and (2) resolve the additional choices w.r.t. the second component of $\mdp_\otimes$.

We introduce functions $\thetas$ and $\thetaa$ that label histories starting with $\fixh$ with elements of $2^\mathcal{C} \cup \{\bot\}$ and define the current state and the current action to pick in $\mdp_\otimes$ in the given history, respectively. Inductively, together with defining the labellings, we also assign one of two distinct tags to these histories, {\em pseudo-accumulating} or {\em pseudo-reaching}.
We will then speak about pseudo-reaching and pseudo-accumulating \emph{phases} which are maximal consecutive ranges within histories labelled so far such that all prefixes in this range are tagged as pseudo-reaching or pseudo-accumulating, respectively. A pseudo-reaching phase is fulfilled if it contains a prefix $h$ in its range such that $\thetas(h)$ is fulfilled.

Initially, we tag the history $\fixh$ of $\mdp$ as pseudo-reaching and set $\thetas(\fixh) = \emptyset$ and $\thetaa(\fixh) = \{(\fixq,\symbopen)\}$.

Suppose that $\thetaa(h)$ has already been defined and the tag of a history $h$ of $\mdp$ has been determined. 

First for an action $a$ and state $t$, we tag the extension $h\cdot a \cdot t$ of $h$ as pseudo-accumulating if (i)
$h$ is tagged as pseudo-accumulating and the length of the current pseudo-accumulating phase is less than $\ell$ so far; or (ii)
$h$ is in a pseudo-reaching phase such that some prefix of $h$ within that phase is fulfilled
and $h$ is good.
Otherwise we tag $h\cdot a\cdot t$ as pseudo-reaching.

Next, we define $\thetas(h\cdot a\cdot t)$ by ``deterministically evolving'' by reading the last state of $h$ as in the definition of $\mdp_\otimes$ at page~\pageref{page:det-evolv}, i.e. $\thetas(h\cdot a\cdot t)$ is the minimal set such that for any $(q,x) \in \thetaa(h)$ there is $(q',x') \in \thetas(h\cdot a\cdot t)$ with $q \tran{\nu(s)} q'$ and $x \stackrel{q'}{\rightsquigarrow} x'$ where the latter relation is the relation from the definition of $\mdp_\otimes$

We define $\thetaa(h\cdot a\cdot t)$ to be the minimum element (w.r.t. set inclusion) of $\mathcal{C}$ satisfying the following
\begin{itemize}
	\item If $h\cdot a\cdot t$ is in a pseudo-accumulating phase, then $\thetaa(h\cdot a\cdot t)$ contains $(q_{in},\symbopen)$ for the initial states $q_{in}$ of $\rabin_{\varphi_i,\fixI}$ for all $i\in \fixI$.
	\item For all $(q,\symbopen)\in \thetas(h\cdot a\cdot t)$ such that for some $(E,F)$,
		$$
		\Pr[h]{\sigma}{\{\omega \mid \occu_{q}(\omega)\cap E = \emptyset, \infi_{q}(\omega) \cap F \neq \emptyset \}} > 1-\lambda
		$$
		we put $(q,(E,F)_\circ)\in \thetaa(h\cdot a\cdot t)$, and otherwise we put $(q,\symbopen)\in \thetaa(h\cdot a\cdot t)$. In the case there are several $(E,F)$
		satisfying the condition above, we pick the least one w.r.t. an arbitrary but a priori fixed total order.
	\item For all $(q,(E,F)_x)\in \thetas(h\cdot a\cdot t)$ we put $(q,(E,F)_x)\in \thetaa(h\cdot a\cdot t)$.
\end{itemize}
Note that the minimum element satisfying these conditions always exists. Also note that these definitions are analogous to those in $\mdp_\otimes$, but in addition we give a rule for ``committing'' to an accepting condition.

Finally, any (finite or infinite) path $\omega = s_0a_0s_1a_1\ldots$ in $\mdp$ initiated in $\fixh$ corresponds to a path
$$\omega_\otimes=(s_0,\thetas(s_0))(a_0,\thetaa(s_0))(s_1,\thetas(s_0a_0s_1))(a_1,\thetaa(s_0a_0s_1))\ldots$$
in $\mdp_\otimes$. Similarly, the strategy $\sigma$ gives rise to a strategy $\sigma_\otimes$ defined, for all $h$, by
$\sigma_\otimes(h_\otimes)(a,\thetaa(h_\otimes)) = \sigma(h)(a)$.
The connection between the labellings and the MDP $\mdp_\otimes$ is completed by the following lemma that can be proven immediately from the definitions.

\begin{lemma}\label{lemma:to-product}
	
	For any set $T$,
	$\Pr[\fixh]{\sigma}{T} = \Pr{\sigma_\otimes}{\{\omega_\otimes\mid \omega\in T\}}$.
\end{lemma}

Note that the strategy $\sigma_\otimes$ in Lemma~\ref{lemma:to-product} is possibly still very complex in
its structure and in particular can reach states with $(q,\bot)$ in the second component. We however show that within a certain finite horizon
this happens with a small probability.

Let $\depth(h)$ be the number of pseudo-accumulating phases along the path $h$ . Let
$T$ be the set of runs $\omega$ that have $\depth(\omega) \ge \ell$, and for which no prefix $h$ with $\depth(h) \le \ell$ has $\thetaa(h) = \bot$.
We will show below that the probability of runs in $T$ is very large. 

\begin{lemma}\label{lemma:alternating-runs}
	$\Pr{\sigma}{T \mid \fixh} \ge 1-3 \cdot \ell^2 \cdot |I|\cdot \lambda = 1-\kappa$. 
\end{lemma}
\begin{proof} 
	First, we start with the set of runs
	$$U = \{\omega \mid \omega\models\psi^\fixI \wedge \bigwedge_{i\in \fixI} \Gf{1} \varphi^\fixI_i\} \cap \fixh$$
	with $\Pr{\sigma}{\Omega \setminus U \mid \fixh} < \lambda$ as given by the assumption of Lemma~\ref{prop:screwing-paths} (here $\Omega$ denotes the set of all runs).
	
	Furthermore, let $V \subseteq U$ be the set of runs where all the ``accumulated'' Rabin automata accept, i.e. runs $\omega$ such that for all $i \in \fixI$ and for all prefixes $h_0$ in an at most $\ell$-th pseudo-accumulating phase, we have that $\rabin_{\varphi_i,I}$ accepts $\omega'$ where $\omega = h_0 \cdot \omega'$.
	For a fixed accumulating phase which starts at some good history $h$, we have (denoting $\sum_{k=0}^{\ell-1} Y_{|h|+k}$ by $Y_{1ph}$ where $|h|$ is the number of states in $h$)
	\begin{align*}
	\ell\cdot(1-\lambda) 
	&< \Ex{\sigma}{Y_{1ph} \mid h}  \\
	&< \ell \cdot \Pr{\sigma}{Y_{1ph} = \ell \mid h} + (\ell - 1) \cdot \Pr{\sigma}{Y_{1ph} < \ell \mid h},
	\end{align*}
	yielding $\Pr{\sigma}{Y_{1ph} = \ell \mid h} > 1-\ell\lambda$. Thus for $Y_{\ell ph}$ denoting the number of Rabin automata accepting in all $\ell$ accumulating phases, we easily obtain $\Pr{\sigma}{Y_{\ell ph} = \ell \cdot \ell \mid \fixh} > 1-\ell^2\cdot \lambda$ and thus $\Pr{\sigma}{U \setminus V} < \ell^2\cdot \lambda$.
	
	For every $i \in I$, we say that \emph{starting after $h$, the history $h'$ decides for an acceptance condition $(E,F)$ of $\rabin_{\varphi_i,\fixI}$} if
	\begin{itemize}
		\item $h$ is in pseudo-accumulating phase, 
		\item $h'$ is the shortest history such that for some $(E',F')$ 
		$$\Pr[h\cdot h']{\sigma}{\{\omega \mid \occu_{q}(\omega)\cap E = \emptyset, \infi_{q}(\omega) \cap F \neq \emptyset \}} > 1-\lambda$$
		where $q$ is the state in which $\rabin_{\varphi_i,\fixI}$ ends after reading $h'$, and
		\item $(E,F)$ is the minimal one among such acceptance conditions $(E',F')$ (w.r.t. the above fixed order).
	\end{itemize}   
	
	We define a set $W \subseteq V$ of runs where this ``decision'' turns out to be correct for all automata started in the first $\ell$ accumulating phases. Technically, $\omega \in W$ if for every $i\in \fixI$ and every splitting $\omega = h\cdot h' \cdot \omega_2$ such that $h$ is in an at most $\ell$-th pseudo-accumulating phase we have the following. If starting after $h$, $h'$ decides for some $(E,F)$, we have $\occu_{q}(\omega_2)\cap E = \emptyset$ and $\infi_{q}(\omega_2) \cap F \neq \emptyset$ where again $q$ is the state in which $\rabin_{\varphi_i,\fixI}$ ends after reading $h'$. 
	
	When starting after a single $h$, $h'$ decides for some $(E,F)$, the probability of not sticking to this decision is by definition at most $\lambda$ (conditioned by $h\cdot h'$). Similarly as before, there are at most $\ell^2 \times |I|$ decisions to take, yielding the overall probability at most $\Pr{\sigma}{V \setminus W} < |I|\cdot \lambda\cdot \ell^2$ of runs that do not stick to decisions up to $\ell$.
	
	For almost every run $\omega \in V$ we have that $\omega \in T$ if $\omega \in W$. Indeed, inductively, for all prefixes $h a t$ of $\omega$ such that $h$ is in an at most $\ell$-th pseudo-accumulating or pseudo-reaching phase and $\thetaa(h) \neq \bot$, we have $\thetaa(h a t) \neq \bot$ because no forbidden state $q$ of a previously decided automaton is visited along any $\omega$ of $W$. Furthermore, every label $(q,(E,F)_\circ)$ is eventually replaced by $(q,(E,F)_\bullet)$ because $\omega \in V$; and every $(q,\symbopen)$ is eventually replaced by some $(q,(E,F)_\circ)$ (for almost every $\omega \in V$) due to Lemma~\ref{lemma:rabin-determined} given below. Thus, the set of labels along $\omega$ becomes at least $\ell$ times fulfilled.
	
	Summing up $\Prb[h]{\sigma}(\Omega \setminus U)$, $\Prb[h]{\sigma}(U \setminus V)$ and $\Prb[h]{\sigma}(V \setminus W)$, we obtain the statement of the lemma.
\end{proof}

An important step in the previous proof was that on almost every accepting path there is a prefix where the Rabin automaton ``decides'' for one accepting condition with high probability.
The proof is again based on Lévy's Zero-One Law.

\begin{lemma}\label{lemma:rabin-determined}
	Let $R$ be a Rabin automaton, $h$ be a path, $V = \{h \cdot \omega \mid R \text{ accepts } \omega \}$, and $\Pr{\sigma}{V} > 0$.
	For almost all $h \cdot \omega' \in V$ there is a prefix $h'$ of $\omega'$ and an acceptance pair $(E,F)$ of $\rabin$ such that	%
	$$\Pr[h\cdot h']{\sigma}{\{\omega \mid \occu_{q}(\omega)\cap E = \emptyset, \infi_{q}(\omega) \cap F \neq \emptyset \}} > 1-\lambda$$
	where $q$ is the state in which $\rabin$ ends after reading $h'$. 
\end{lemma}
\begin{proof}
	Let the acceptance conditions of $R$ be $(E_i,F_i)_{1 \leq i \leq n}$ and its initial state be $q_0$. 
	For each $i$, let $R_i$ be the set $$\{h\cdot \omega' \in V \mid 
	\infi_{q_0}(\omega') \cap E_i = \emptyset,
	\infi_{q_0}(\omega') \cap F_i \neq \emptyset \}$$ and $\mathbf{1}_{R_i}$ be its indicator function.
	As for each $h\cdot \omega'$ in some 
	$R_i$, 
	$\idf_{R_i}(h\cdot \omega') = 1$, we also have from Lemma~\ref{lemma:levy} (Lévy's zero one law) that $\lim_{k\to\infty} \Pr{\sigma}{R_i \mid h_k} = 1$ where $h_k$ are the prefixes of $h\cdot \omega'$ of length $k$. 
	Hence, there is $k$ such that all prefixes $h_{k'}$ for $k'\geq k$ satisfy:
	\begin{align}\label{eq:rabin-det1}
	\Pr{\sigma}{ R_i \mid h_k  } 
	\;>\;
	1-\frac{\lambda}{2}.
	\end{align}
	
	Let us fix an arbitrary partition of $V$ into \emph{disjoint} sets $R'_1, \ldots, R'_n$ such that for all $1 \leq i \leq n$, $R'_i \subseteq R_i$. 
	For each $i$ and run $h\cdot \omega' = s_0 a_0 \cdots $ let $$lastSin_i(\omega') = \sup(\{0\} \cup \{n \mid s_n \in E_i \}).$$

	Let $h\cdot \omega' \in R'_i$. 
	As we have $\idf_{\{lastSin_i > lastSin_i(h\cdot \omega' )\}}(h\cdot \omega' ) = 0$, we also have from Lemma~\ref{lemma:levy} that $\lim_{k\to\infty} \Pr{\sigma}{\{lastSin_i > lastSin_i(h\cdot \omega')\} \mid h_k} = 0$.  Hence, there is $k \in \Nset$ such that $k > lastSin_i(h\cdot \omega' )$ and all prefixes $h_{k'}$ of length $k' \geq k$ satisfy
	\begin{align}\label{eq:rabin-det2}
	\Pr{\sigma}{ \{lastSin_i > k\} \mid h_k  } 
	\;<\;
	\frac{\lambda}{2}.
	\end{align}

	\noindent
	In total, we obtain from~\ref{eq:rabin-det1} and~\ref{eq:rabin-det2} the desired statement.
\end{proof}

Before constructing the accumulating and reaching strategies, we state the following  lemmas that we will need.

The first lemma says that if we can achieve a certain event with a large enough probability in an MDP, then we can achieve it with probability $1$. The proof follows from the fact that there are optimal deterministic strategies with memory of size 2.
\begin{lemma}\label{lemma:pone}
	Let $\mdp$ be an MDP with state space $S$, where each state $s$ is labelled with an atomic proposition $s$ unique to this state, let $p$ be the minimal probability occurring in it, and let $G_1$ and $G_2$ be two sets of states. The following statements hold true:
	\begin{enumerate}
		\item  If $\sup_\sigma \Pr{\sigma}{\F (G_1 \land \F G_2)} > 1 - p^{2\cdot|S|}$, then $\Pr{\sigma}{\F (G_1 \land \F G_2)}= 1$ for some $\sigma$.
		\item  If $\sup_\sigma \Pr{\sigma}{\{\omega = s_0a_0s_1a_1\ldots \mid \forall i\le |S| : s_i \in G_1 \}} > 1 - p^{|S|}$, then $\Pr{\sigma}{\G G_1)}= 1$ for some $\sigma$.
	\end{enumerate}
\end{lemma}
\begin{proof}
	Let us analyse the second case which is slightly more technical.
	The set $\{\omega = s_0a_0s_1a_1\ldots \mid \forall i\le |S| : s_i \in G_1 \}$ can be captured using an LTL property
	and so the supremum is realised by some deterministic strategy $\sigma'$. Suppose it is lower than $1$.
	Then, since $\sigma'$ is deterministic, there must be a history $s_0 a_0 s_1 a_1\ldots s_i$ for $i\le |S|$
	such that $s_i\not\in G_1$ and $\Pr{\sigma'}{s_0 a_0 s_1 a_1\ldots s_i} \ge p^{i}$, which is a contradiction.
	
	The first case can be proved similarly, we only need to consider that deterministic strategies with memory of size 2 are
	sufficient to achieve the supremum.
\end{proof}

From now on, we will consider the strategy $\sigma_\otimes$ in $\mdp_\otimes$ instead of $\sigma$. We transfer the labelling with a pseudo-reaching and pseudo-accumulating phase to runs of $\mdp_\otimes$ in the straightforward way.

Let $\hat W_i$ be the set of histories that are in $i$-th pseudo-accumulating phase and whose predecessors are not in $i$-th pseudo-accumulating phase.
In order to define accumulating and reaching strategies, we need to select subsets of these histories that are ``connected'' with high probability.
We thus select non-empty sets $W_i\subseteq \hat W_i$ which in addition satisfy 
\begin{align*}
\Pr{\sigma_\otimes}{W_{i}} \ge 1-4 {\cdot} \kappa^{1/2^{\ell - i}} & \quad  \text{for all $1{\le} i{\le} \ell$, and } \\
\Pr{\sigma_\otimes}{W_{i+1}\mid h} \ge 1-4 {\cdot} \kappa^{1/2^{\ell - i}} & \quad \text{for all $h{\in}W_i$ and $1{\le} i {<} \ell$}
\end{align*}
(recall that we interchangeably interpret a set of histories also as a set of runs starting with some history from the set).
This is indeed possible, it suffices to put $W_{\ell} = \hat W_{\ell}$ that satisfies the first condition by Lemma~\ref{lemma:alternating-runs}. Supposing $W_{i+1}$ has been defined, we get $W_i$ by the following.

\begin{lemma}\label{lemma:future-prob}
	Let $\eps > 0$ and $\mathbb{P}$ be a probability measure. Further, let $U$ be a set of runs such that $\P{U} \ge 1-\varepsilon$, and let $H$ be the prefix-free set of finite paths such that $\P{H} \ge 1-\varepsilon$. There is a set $V\subseteq H$ with $\P{V} \ge 1 - 2\sqrt\eps$
	and $\P{U \mid w} > 1-\sqrt\eps$ for all $w\in V$.
\end{lemma}
\begin{proof}
	We can assume $\eps < 1/4$, otherwise the statement holds for any set $V$ as $2\sqrt{\eps} \geq 1$.
	
	We set $V := \{y\in H\mid \P{U\mid y} > 1-\sqrt\varepsilon\}$. We claim that $\P{V} > 0$. Let us assume the contrary. This means that all $y \in H$ with $\P{y} > 0$ satisfy $\P{U\mid y} \leq 1 - \sqrt{\eps} \leq 1- 2\eps \leq \P{U} - \eps$ and hence, also $\P{U \mid H} \leq \P{U} - \eps$. We have 
	$$\P{U} \leq \P{H} \cdot \P{U \mid H} + \P{\bar H} \cdot 1 \leq (1-\eps)\cdot (\P{U} - \eps) + \eps = \P{U} - \eps \cdot \P{U}$$
	yielding a contradiction which proves that $\P{V} > 0$.

	As all runs can be partitioned into sets $\bar H$, $V$, and $H\setminus V \subseteq \bar V$, we have
	\[
	1-\varepsilon \le \P{U} \leq \P{\bar H} \cdot 1 + \P{V}\cdot 1 + \P{\bar V} \cdot \P{U\mid \bar{V}}
	\]
	where $1$ overapproximates (potentially undefined) probabilities $\P{U \mid \bar H}$ and $\P{U\mid V}$.
	Since $\P{\bar H} < \eps$ and $\P{U\mid \bar V} \leq 1-\sqrt\eps$, we obtain
	\begin{align*}
	1- 2\varepsilon &\le \P{V} + \P{\bar V} \cdot (1-\sqrt\eps) \\
	1- 2\varepsilon &\le \P{V} + (1 - \P{V}) \cdot (1-\sqrt\eps) \\
	\sqrt\eps -2\eps &\le \P{V} \sqrt\eps
	\end{align*}
	and so $\P{V} \geq 1 - 2\sqrt{\eps}$.
\end{proof}

\noindent
In fact, it suffices to set $U = W_{i+1}$ and $H = \hat{W}_i$ and obtain $W_i$ as $V$ from the lemma. The probability of all such $\hat{W}_i$ is $\geq 1-\kappa$. It is easy to prove by induction that the probability of each $W_i$ is $1-2^{a_i} \cdot \kappa^{1/2^{\ell - i}}$ where $a_i = \sum_{j=i}^{\ell-1} 1/(\ell-i)$ obtaining the first inequality as $\lim_{i\to\infty} a_i = 4$. The second inequality is guaranteed by the properties of $V$ by Lemma~\ref{lemma:future-prob}.

Even the sets of $W_i$ are still not enough for our proof, we would like to get sets of histories that are ``connected'' with high probability from anywhere \emph{within} the accumulating phase.
For all $i$ and all $h\in W_i$ we apply the following lemma and obtain a prefix-free set of paths $Z_h$ such that
$\Pr{\sigma_\otimes}{Z_h} \ge 1 - 4 \cdot \ell^{3/2} \cdot \kappa^{1/2^\ell}$ and for all prefixes $h'$ of any path in $Z_h$ we have $\Pr{\sigma_\otimes}{W_{i+1}\mid h'} \ge 1 - 2 \cdot \kappa^{1/2^\ell}$.

\begin{lemma}\label{lemma:still-in-set}
	Let $W$ be a set of runs such that
	\[
	\Pr{\sigma_\otimes}{W} > 1-\varepsilon
	\]
	then there is a prefix-free set $V$ of finite paths of length
	$\ell$ such that $\Pr{\sigma_\otimes}{V } \ge 1- 2\cdot\ell\cdot\sqrt{\ell\epsilon}$
	and for all prefixes $h'$ of a path in $V$
	we have
	\[
	\Pr{\sigma_\otimes}{W \mid h'} > 1 - \sqrt{\epsilon}
	\]
\end{lemma}
\begin{proof}
	For all $k$, we can find a set
	$V_k$ of paths of length $k$ such that $\Pr{\sigma_{\otimes}}{V_k} \ge 1 - 2\sqrt{\epsilon}$ and
	$\Pr{\sigma_\otimes}{W \mid h'} > 1 - \sqrt{\epsilon}$ for all $h'\in V_k$; this is possible by
	Lemma~\ref{lemma:future-prob}, where for $H$ we take all paths
	of length $\ell$. The set $V$ is then
	obtain $V = \bigcap_{i=1}^\ell V_i$ (note that this is indeed a set of paths).
\end{proof}

Finally, we are ready to obtain the accumulating and reaching strategies. Below, $\last{h}$ is the last state of a path $h$.

\begin{lemma}
	For $i\le \ell$, $h\in W_i$ and $h'\in Z_h$, there is a strategy $\bar\zeta_{h'}$
	that from $\last{h'}$ almost surely reaches $\{\last{h''} \mid h''\in W_{i+1}\}$ after passing through a fulfilled state.
\end{lemma}
\begin{proof}
	Such strategy always exists because of Lemma~\ref{lemma:pone} and because
	$\Pr{\sigma_\otimes}{W_{i+1}\mid w} \ge 1 - 2 \cdot \kappa^{1/2^\ell}$ by properties of elements of $W_i$ and $Z_w$.
\end{proof}

The following lemma can be easily obtained from Lemma~\ref{lemma:pone}.
\begin{lemma}
	For all $w\in W_i$, there is a memoryless deterministic strategy $\bar\pi_w$ (in $\mdp_\otimes$) which, when started in $\last{w}$, only ever reaches
	states and uses actions that occur on some history of $Z_w$.
\end{lemma}

In addition, denote by $\pi_{s,C}$
a strategy $\bar\pi_w$ for $w$ belonging to $W_i$ for $i = \min \{j \mid \exists w'\in W_j, \last{w'} = (s,C)\}$.
By $(M_{s,C},N_{s,C})$ we denote the tuple of sets of states and actions that $\pi_{s,C}$ visits when started in $(s,C)$.

Let $\zeta_{s,C}$ be a strategy $\bar\zeta_w$ where $w\in \bigcup_{w'\in W_{i}} Z_{w'}$ for $i= \min \{j \mid \exists w'\in \bigcup_{w'\in W_{j}} Z_{w'} \wedge \last{w'} = (s,C)\}$. Let $\rank{(s,C)} = i$.

We inductively build $(M,N)$ as follows. Initially, $M = \{\last{w} \mid w\in W_1\}$ and $N=\emptyset$. We then keep adding to $M$ and $N$, until a fixpoint is reached, (i) the states and actions of $(M_{s,C},N_{s,C})$ for all $(s,C)\in M$,
and (ii) last states of histories $W_{i+1}$ for $i$ such that there is $(s,C)\in M$ with $\rank{(s,C)} = i$.
We claim that this procedure is well-defined in the sense that
the sets $W_{i+1}$ in step (ii) above were always defined, i.e. that $i<\ell$ in every case.
For this, we need to show that whenever $H_{(s,C)}$ is taken in the definition,
then $\rank{(s,C)} \le \ell-1 = |S_\otimes|$. Letting
$\rank{M} = \max\{\rank{(s,C)} \mid (s,C)\in M\}$, we can argue that initially $\rank{M} = 1$ and with every iteration of steps (i) and (ii) the rank increases at most by 1.
Since only $|S_\otimes|$ elements can be added to $M$ before a fixpoint is reached, we get that the bound on $\rank{M}$ is $|S_{\otimes}|$.

Now we claim that the Lemma~\ref{prop:screwing-paths} is satisfied.
\begin{itemize}
 \item As for the property of $N$, note that
we were only adding states to $N$ if they were last states of a history in a pseudo-accumulating phase, and by definition of $\thetaa$ we have
$(q_{in}, \symbopen)$ in the second component of such states for the initial states $q_{in}$ of the automata $\rabin_{\varphi_i, I}$ for all $i\in I$.
 \item For item \ref{item:screwing-a}, the strategy $\pi$ is defined as follows. Let $h$ be a history starting in $(s,C)\in M$, we put $\pi(h) = \pi_{(s',C')}(h)$ where $(s',C')$ is the element such that $(s,C)\in M_{s',C'}$. For any other history we define $\pi$ arbitrarily.
 \item For item \ref{item:screwing-c}, the strategy $\zeta$ is defined as follows. Let $h$ be a history starting in $(s,C)\in M$, we put $\zeta(h) = \zeta_{(s,C)}(h)$. For any other history we define $\zeta$ arbitrarily.
\end{itemize}

\subsection{Details for proof of Theorem~\ref{thm:main-mdp} and Section~\ref{sec:alg}}

\begin{reftheorem}{thm:main-mdp}
	The controller-synthesis problem for 1-fLTL for MDPs is solvable in time polynomial
	in the size of the model and doubly exponential in the size of the formula.
\end{reftheorem}

\begin{proof}
        We now give a more detailed description of the algorithm that is presented in Section~\ref{sec:alg}.
	
	\begin{enumerate}
		\item Construct the automata $\rabin_{\xi,I}$ for all $\xi\in \{\psi,\varphi_1,\ldots,\varphi_n\}$ and $I\subseteq \{1,\ldots,n\}$.
		\item Initialize $\Upsilon := \emptyset$.
		\item Repeat the following for every $I$. Find the largest sets $(M,N)$ satisfying the conditions \ref{item:screwing-a}--\ref{item:screwing-c} of Lemma~\ref{prop:screwing-paths}.
		It can be done as follows:
		\begin{itemize}
			\item Let $\trunc{M,N}$ denote the tuple $(M',N')$ that contains maximal subsets of $M$ and $N$
			satisfying that for every $s\in M$ there is $a\in N$ such that $\Delta(s,a)$ is defined and for every $s'$
			contained in the support of $\Delta(s,a)$ we have $s'\in M$. (Easily obtained by iteratively pruning actions and states violating the conditions.)
			\item We start with $M=S_\otimes$ and $N$ containing all ``accumulating'' actions $(a,C)$ with $\{(q_{in}, \symbopen) \mid q_{in} \text{ is the initial state of $\rabin_{\varphi_i,\fixI}$ for } i \in \fixI\} \subseteq C$. Then we apply the following steps until a fixpoint is reached:
			\begin{itemize}
				\item[(a)] $(M,N) := \trunc{M,N}$;
				\item[(b)] Remove from $M$ all states that do not satisfy item~\ref{item:screwing-a} or item~\ref{item:screwing-c} of
				Lemma~\ref{prop:screwing-paths}. (Easily achieved by qualitative safety and reachability analysis in $\mdp_\otimes$.)
			\end{itemize}
		\end{itemize}
		This yields a set $(M,N)$, and we add to $\Upsilon$ all pairs $(s,q)$ such that
		$(s,\{(q,\symbopen)\}) \in M$.
		\item Compute an optimal strategy $\sigma'$ for ``reaching'' $\Upsilon$ (defined in Proposition~\ref{prop:reach}) and return the probability that $\sigma'$ ``reaches'' $\Upsilon$. It can be done as follows:
		\begin{itemize}
			\item By $\mdp_\Diamond$ we denote the ``naive'' product of $\mdp$ with all the main Rabin automata $\rabin_{\psi,I}$
			for all $I\subseteq \{1,\ldots,n\}$. Formally, fixing $I_0,\ldots, I_m$ an enumeration of subsets of $\{1,\ldots,n\}$,
			the state space $\states_\Diamond$ of $\mdp_\Diamond$ contains tuples
			$(s,q^{I_0},\ldots,q^{I_m})$ where $q^{I_j}$ is a state of $\rabin_{\psi,I_j}$,
			the set of actions is $\actions_\Diamond = \actions$, and the transition function $\tra_\Diamond$
			is given by
			\[
			\tra_\Diamond((s,q^{I_0},\ldots, q^{I_m}),a)(t,\bar q^{I_0},\ldots \bar q^{I_m}) = \tra(s,a)(t)
			\]
			when for every $0 \leq i \leq m$, we have $q^{I_i}\tran{s} q'^{I_i}$.
			
			\item Furthermore, let $\Upsilon_\Diamond\subseteq S_\Diamond$ be the set of all $(s,q^{I_0},\ldots,q^{I_m})$ s.t.
			there is $i$ with $(s,q^{I_i})\in \Upsilon$.
			
			\item By the construction, we easily obtain equivalence of strategies of the following form.
			For any $\sigma$ in $\mdp$ there is $\sigma_\Diamond$ in $\mdp_\Diamond$, and also for any $\sigma_\Diamond$ there is $\sigma$ such that 
			$$\Pr{\sigma}{\{ \omega \mid \text{$\omega$ reaches a pair from $\Upsilon$} \}} = \Pr{\sigma_\Diamond}{\{ \omega \mid \text{$\omega$ reaches some state from $\Upsilon_\Diamond$} \}}.$$
			Let us prove the statement. For any \emph{finite or infinite} path $\omega=s_0a_0s_1a_1\ldots$ in $\mdp$ initiated in $\sinit$
			there is a unique path $\omega_\Diamond = (s_0,q_0^{I_0},\ldots, q_0^{I_m})a_0(s_1,q_1^{I_0},\ldots, q_1^{I_m})a_1\ldots$ with $s_i = s'_i$ and $a_i=a'_i$ for all $i$. 
			For a fixed $\sigma_\Diamond$ we define $\sigma$ by $\sigma(h) = \sigma_\Diamond(h_\Diamond)$ for all $h$.
			Similarly, for a fixed $\sigma$, we define $\sigma_\Diamond$ by $\sigma_\Diamond(h_\Diamond) = \sigma(h)$ for all $h$.
			The equality easily follows from the definitions.
			
			\item The above statement allows us to compute an optimal strategy $\sigma_\Diamond$ in $\mdp_\Diamond$ using ordinary reachability algorithms and set $\sigma'$ to the corresponding strategy $\mdp$.
		\end{itemize}
	\end{enumerate}
	
	\noindent
	
	Let us now analyse the complexity of the algorithm in more detail. Each of the Rabin automata in step 1. above can be computed in time $2^{2^{\poly{|\varphi|}}}$,
	and since there is exponentially many such automata (in $|\varphi|$), step 1. takes time $2^{2^{poly(|\varphi|)}}$.
	In step 3., for a fixed $I$, $M$ and $N$ the result of $\trunc{M,N}$ can be computed in polynomial time in the size of $M$ and $N$; the same holds 	for satisfaction of the conditions in (b).
	The size of $M\subseteq S_\otimes$ and
	$N\subseteq A_\otimes$ is
	$\poly{S}\cdot 2^{2^{\poly{|\varphi|}}}$, and for a fixed $I$ the fixpoint is reached in at most $|S_\otimes| \cdot |A_\otimes|$ iterations.
	Moreover, there is
	at most $2^{|\varphi|}$ different $I$s. Hence, step 3. can be performed in time $\poly{S}\cdot 2^{2^{\poly{|\varphi|}}}$.
	Finally, in step 4. we are computing reachability probability in a MDP $\mdp_\Diamond$ which is of size $\poly{S} \cdot 2^{2^{\poly{|\varphi|}}}$,
	and so also this step can be done in time $\poly{S}\cdot 2^{2^{\poly{|\varphi|}}}$.
	This completes proof of Theorem~\ref{thm:main-mdp}.
\end{proof}

\end{document}